\documentclass[journal,draftcls,onecolumn,twoside]{IEEEtran}

\usepackage{amsmath,amssymb,amsthm}
\usepackage{enumerate}
\usepackage{stfloats}
\usepackage{comment}
\usepackage{bbm}
\usepackage{siunitx}
\usepackage{mathtools}
\usepackage{accents,latexsym,cancel}
\usepackage{cite}

\let\labelindent\relax
\usepackage{enumitem}
\usepackage[dvipsnames]{xcolor}
\definecolor{myPink}{RGB}{255,105,183}

\usepackage[T1]{fontenc}
\usepackage{graphics}
\usepackage{epsfig}
\usepackage[mathscr]{euscript}
\usepackage{algorithm}
\usepackage[noend]{algpseudocode}
\usepackage{bbm}
\makeatletter
\def\BState{\State\hskip-\ALG@thistlm}
\makeatother

\DeclareMathAlphabet{\mathpzc}{OT1}{pzc}{m}{it}

\usepackage[hidelinks]{hyperref}

\usepackage{tikz}
\usetikzlibrary{arrows,shapes,chains,matrix,positioning,scopes,patterns,calc,
decorations.markings,
decorations.pathmorphing,
}

\usepackage{pgfplots}
\pgfplotsset{compat=1.3}
\usepgflibrary{shapes}

\newtheorem{theorem}{Theorem}
\newtheorem{lemma}[theorem]{Lemma}
\newtheorem{proposition}[theorem]{Proposition}
\newtheorem{definition}[theorem]{Definition}

\newtheorem{remark}[theorem]{Remark}
\newtheorem{corollary}[theorem]{Corollary}
\newtheorem{assumption}[theorem]{Assumption}

\renewcommand{\epsilon}{\varepsilon}

\newcommand{\norm}[1]{\lVert #1\rVert
}

\newcommand{\abs}[1]{\left\lvert#1\right\rvert}

\newcommand{\RNum}[1]{\uppercase\expandafter{\romannumeral #1\relax}}

\newcommand{\rv}{\ensuremath{\mathbf{r}}}

\newcommand{\uv}{\ensuremath{\mathbf{u}}}
\newcommand{\vv}{\ensuremath{\mathbf{v}}}
\newcommand{\wv}{\ensuremath{\mathbf{w}}}
\newcommand{\xv}{\ensuremath{\mathbf{x}}}
\newcommand{\yv}{\ensuremath{\mathbf{y}}}
\newcommand{\zv}{\ensuremath{\mathbf{z}}}

\newcommand{\Av}{\ensuremath{\mathbf{A}}}
\newcommand{\Bv}{\ensuremath{\mathbf{B}}}
\newcommand{\Cv}{\ensuremath{\mathbf{C}}}
\newcommand{\Ev}{\ensuremath{\mathbf{E}}}
\newcommand{\Fv}{\ensuremath{\mathbf{F}}}
\newcommand{\Gv}{\ensuremath{\mathbf{G}}}
\newcommand{\Hv}{\ensuremath{\mathbf{H}}}
\newcommand{\Iv}{\ensuremath{\mathbf{I}}}
\newcommand{\Kv}{\ensuremath{\mathbf{K}}}
\newcommand{\Lv}{\ensuremath{\mathbf{L}}}
\newcommand{\Mv}{\ensuremath{\mathbf{M}}}
\newcommand{\Nv}{\ensuremath{\mathbf{N}}}
\newcommand{\Ov}{\ensuremath{\mathbf{O}}}
\newcommand{\Pv}{\ensuremath{\mathbf{P}}}
\newcommand{\Qv}{\ensuremath{\mathbf{Q}}}
\newcommand{\Rv}{\ensuremath{\mathbf{R}}}
\newcommand{\Sv}{\ensuremath{\mathbf{S}}}
\newcommand{\Tv}{\ensuremath{\mathbf{T}}}
\newcommand{\Uv}{\ensuremath{\mathbf{U}}}
\newcommand{\Vv}{\ensuremath{\mathbf{V}}}
\newcommand{\Wv}{\ensuremath{\mathbf{W}}}
\newcommand{\Xv}{\ensuremath{\mathbf{X}}}
\newcommand{\Yv}{\ensuremath{\mathbf{Y}}}
\newcommand{\Zv}{\ensuremath{\mathbf{Z}}}

\newcommand{\du}{\ensuremath{\underline{d}}}

\newcommand{\iu}{\ensuremath{\underline{i}}}
\newcommand{\ju}{\ensuremath{\underline{j}}}

\newcommand{\Cb}{\ensuremath{\mathbb{C}}}

\newcommand{\Eb}{\ensuremath{\mathbb{E}}}
\newcommand{\Fb}{\ensuremath{\mathbb{F}}}

\newcommand{\Rb}{\ensuremath{\mathbb{R}}}

\newcommand{\cC}{\ensuremath{\mathcal{C}}}

\newcommand{\cG}{\ensuremath{\mathcal{G}}}

\newcommand{\cN}{\ensuremath{\mathcal{N}}}

\newcommand{\cS}{\ensuremath{\mathcal{S}}}

\newcommand{\pzf}{\ensuremath{\mathpzc{f}}}
\newcommand{\pzg}{\ensuremath{\mathpzc{g}}}

\newcommand{\pzi}{\ensuremath{\mathpzc{i}}}

\newcommand{\zerov}{\ensuremath{\mathbf{0}}}

\newcommand{\Sigmav}{\ensuremath{\boldsymbol{\Sigma}}}

\newcommand{\Phiv}{\ensuremath{\boldsymbol{\Phi}}}

\newcommand{\xiv}{\ensuremath{\boldsymbol{\xi}}}
\newcommand{\zetav}{\ensuremath{\boldsymbol{\zeta}}}

\newcommand{\IDm}{\ensuremath{\mathbf{I}}}

\newcommand{\transpose}{\ensuremath{\mathsf{T}}}
\newcommand{\hermitian}{\ensuremath{\mathsf{H}}}

\DeclareMathAlphabet{\mcl}{OMS}{cmsy}{m}{n}

\newlength\tikzwidth
\newlength\tikzheight

\newcommand{\defeq}{\triangleq}

\definecolor{mycolor1}{rgb}{0.63529,0.07843,0.18431}
\definecolor{mycolor2}{rgb}{0.00000,0.44706,0.74118}
\definecolor{mycolor3}{rgb}{0.00000,0.49804,0.00000}
\definecolor{mycolor4}{rgb}{0.87059,0.49020,0.00000}
\definecolor{mycolor5}{rgb}{0.00000,0.44700,0.74100}
\definecolor{mycolor6}{rgb}{0.74902,0.00000,0.74902}

\begin{document}

\title{Complex Approximate Message Passing\\with Non-separable Denoising}
\author{Vishnu Teja Kunde, Alessandro Mirri,~\IEEEmembership{Graduate Student Member,~IEEE}, \\
Jean-Francois Chamberland,~\IEEEmembership{Senior Member,~IEEE}, and Enrico Paolini,~\IEEEmembership{Senior Member,~IEEE}
\thanks{Vishnu Teja Kunde and Jean-Francois Chamberland are with the Department
of Electrical and Computer Engineering, Texas A\&M University, College Station,
TX, 77843 USA (emails: \{vishnukunde,~chmbrlnd\}@tamu.edu). \emph{Corresponding author: Vishnu Teja Kunde.}}
\thanks{Alessandro Mirri and Enrico Paolini are with the CNIT/WiLab, DEI, University of Bologna, Bologna, Italy (emails: \{alessandro.mirri7,~e.paolini\}@unibo.it).}}

\maketitle

\begin{abstract}
Approximate Message Passing (AMP) is a general framework for iterative algorithms, originally developed for compressed sensing and later extended to a wide range of high-dimensional inference problems.
Although recent work has advanced matrix AMP, complex AMP, and AMP for non-separable functions independently, a unified state evolution theory for complex AMP with non-separable denoisers has been lacking.
This article fills that gap by establishing state evolution in the setting of complex, non-separable denoising functions.
The proposed approach constructs an augmented real-valued system that lifts the problem to a higher-dimensional space, then recovers the complex domain through a many-to-one canonical transformation.
Under this construction, the Onsager correction naturally involves Wirtinger derivatives, and the resulting state evolution reduces to scalar complex recursions despite the non-separable structure of the denoisers.
The framework extends to the matrix-valued setting, accommodating multiple feature vectors simultaneously.
This generalization enables AMP to exploit joint structural constraints, such as simultaneous group and element sparsity, in complex-valued recovery problems.
The complex sparse group least absolute shrinkage and selection operator (LASSO) serves as a key instantiation, motivated by preamble detection in Orthogonal Time-Frequency Space (OTFS)-based unsourced random access.
Numerical experiments confirm that state evolution accurately predicts performance and show that complex non-separable denoising can produce significant gains over separable and real-valued alternatives.
\end{abstract}

\begin{keywords}
Approximate Message Passing (AMP), Complex AMP, State Evolution, Sparse Group LASSO, Non-Separable Denoisers, Wirtinger Derivatives.
\end{keywords}

\section{Introduction}
\label{sec:intro}

Approximate Message Passing (AMP) has emerged as a powerful iterative framework for solving high-dimensional inference problems with random measurement matrices~\cite{donoho2009message,bayati2011dynamics,bayati2015universality}.
The framework offers computationally efficient algorithms and a precise asymptotic characterization of performance through \emph{state evolution} (SE).
AMP has been successfully applied in compressed sensing, statistical signal processing~\cite{parker2014bilinearI,parker2014bilinearII,metzler2016denoising,rangan2019vector}, high-dimensional inference \cite{tan2023mixed,ivkov2025fast}, and
communication~\cite{barbier2017approximate,rush2017capacity,amalladinne2021unsourced}.
The SE formalism enables sharp predictions of algorithmic performance, which can in turn be used for parameter selection.
While the original proofs of SE relied on Gaussian measurement matrices and separable denoising functions, subsequent work has progressively extended the theory to more general settings.
These include non-separable functions~\cite{berthier2020state}, matrix-valued variables~\cite{donoho2013information,feng2022unifying}, and structured sensing operators~\cite{gerbelot2023amp}, substantially enlarging the scope of AMP algorithms.

Complex-valued signal processing plays a central role in communications and related fields.
Many physical systems, particularly in wireless communications, are naturally modeled in the complex domain.
For example, orthogonal frequency-division multiplexing (OFDM) and orthogonal time-frequency space (OTFS) modulation rely on representing information-bearing signals as complex waveforms.
Similar complex-valued models appear in radar, sonar, and array processing, where phase coherence across sensors carries critical information.

Because of this intrinsic structure, algorithms designed for real-valued signals often fail to capture the statistical dependencies, rotational invariances, or phase relationships that are central to complex-valued data.
This has motivated the development of complex-signal processing frameworks that treat real and imaginary components in a unified manner rather than as two independent channels~\cite{Maleki2013CAMP,chen2018sparse,ma2019optimization,burak2025jointmsgdetection,sterk2025iterative}.
These developments lie at the intersection of high-dimensional inference and complex-valued communications theory, highlighting the potential benefits of complex AMP and SE applied to practical systems.

This article extends the complex AMP framework to non-separable denoising functions.
A key motivation is the \emph{sparse group LASSO} (SGL) problem~\cite{simon2013sparsegrouplasso,yuan2006grouplasso}, which combines group-level sparsity (via a group penalty) and element-level sparsity (via an $\ell_1$ penalty within each group).
The resulting proximal operator is inherently non-separable, as the group threshold depends on the joint norm of multiple components.
This structure arises naturally in preamble detection under OTFS modulation~\cite{mirri2025approximate}, where active users form sparse groups with sparse channel taps.

A naive approach would reduce the complex system to a real one via the isomorphism $\Cb^n \cong \Rb^{2n}$.
However, both natural embeddings, stacking rows or concatenating columns, either destroy the independent and identically distributed (i.i.d.) structure of the sensing matrix required for AMP, or decouple the real and imaginary components, precluding joint complex processing.
The proposed approach instead constructs an augmented real matrix-valued system~\cite{javanmard2013state,gerbelot2023amp} and recovers the complex formulation through a many-to-one canonical transformation.

This line of development builds on several key contributions.
Maleki et al.~\cite{Maleki2013CAMP} introduced AMP for scalar (separable) complex denoisers; \c{C}akmak, Gkiouzepi, Opper, and Caire~\cite{burak2025jointmsgdetection} recently extended complex AMP algorithms to row-wise \emph{separable} matrix denoisers by combining the \emph{Householder Dice} representation~\cite{lu2021householder} with a complex variant of Stein's lemma.
The present work builds on the graph-based AMP framework of Gerbelot and Berthier~\cite{gerbelot2023amp}, extending it to the complex domain with fully \emph{non-separable} denoisers.
The canonical transformation extends naturally to the matrix setting, where $p$ feature vectors are estimated jointly (Section~\ref{sec:matrix-extension}).

\noindent
\textbf{Notations:}
Let $\Rb^n$, $\Cb^n$ be $n$-dimensional vector spaces over the real and complex numbers.
The operators $(\cdot)^\transpose$, $(\cdot)^\hermitian$ denote transpose and Hermitian, respectively.
A matrix $\Sv$ is \emph{Hermitian positive definite} if $\Sv = \Sv^\hermitian$ and $\xv^\hermitian \Sv \xv > 0$ for all nonzero $\xv$ and a matrix $\Cv$ is \emph{complex symmetric} if $\Cv = \Cv^\transpose$.
We let $\pzf$, $\pzg$ denote complex vector-valued functions, and we write $f$, $g$ for real matrix-valued functions.
The $\iota$th column and the $\iota$th row of matrix $\Av$ are denoted by $\Av_{:, \iota}$ and $\Av_{\iota, :}$, respectively.
The Jacobian matrix $\left(\frac{\partial\uv}{\partial\vv} \right)$ of vector-valued variable $\uv$ with respect to $\vv$ has entries $\left(\frac{\partial\uv}{\partial\vv} \right)_{i,j} = \frac{\partial \uv_i}{\partial \vv_j}$.
For $\vv \in \Cb^n$, the empirical mean is $\langle \vv \rangle = \frac{1}{n} \sum_{\iota = 1}^n \vv_\iota$.
We use $\pzi = \sqrt{-1}$ and let $\rho = [ 1 , \pzi ]^\transpose$.
For a function $\pzg:\Cb^n\to \Cb^n$, we adopt the notation $(\pzg)'$, given by
\begin{equation*}
(\pzg)'_\iota(\vv) = \frac{\partial \pzg_\iota(\vv)}{\partial \vv_\iota} = \frac{1}{2} \left( \frac{\partial}{\partial \Re \vv_\iota} - \pzi \frac{\partial}{\partial \Im \vv_\iota}  \right) (\pzg_\iota (\vv))
\end{equation*}
for the Wirtinger derivatives.
For $\vv_1, \vv_2 \in \Cb^n$, we have $\langle \vv_1, \vv_2 \rangle = \vv_2^\hermitian \vv_1 = \sum_\iota v_{1,\iota} \overline{v_{2,\iota}} \in \Cb$, where $\overline{(\cdot)}$ denotes complex conjugation. We write $X_n  \stackrel{\mathrm{P}}{\simeq} Y_n$ to mean that $X_n - Y_n \stackrel{\mathrm{P}}{\to} 0$ (in probability) as $n\to \infty$.

\section{Motivation and Formulation}
\label{sec:motivation}

Consider the canonical linear inverse problem
\begin{equation} \label{eq:InverseProblem}
\yv = \Av \xv_0 + \wv
\end{equation}
where $\xv_0 \in \Fb^{N}$ is the signal vector, $\Av \in \Fb^{n\times N}$ is the measurement matrix, and $\wv \in \Fb^{n}$ is additive Gaussian noise with i.i.d. entries.
The underlying field $\Fb$ is either $\Rb$ or $\Cb$.
Let $\delta = n/N$ denote the undersampling ratio.
We are especially interested in the setting where $n$, $N$ are large and the entries of $\Av$ are i.i.d. Gaussian.
The task is to estimate $\xv_0$ from (noisy) observation vector $\yv$.
As part of this inference task, we know that $\xv_0$ is a structured vector and we wish to derive an iterative algorithm that leverages this information.

The quintessential example of such a structure is the well-studied sparse recovery problem, which can be solved using a proximal formulation and the least-absolute shrinkage and selection operator (LASSO)~\cite{donoho2009message}.
A natural extension replaces element-wise sparsity with group sparsity.
One can approach this problem using the group LASSO (GL) formulation,
\begin{equation} \label{eq:complex-group-lasso}
\textstyle
\min_{\xv \in \Fb^N} \lVert \yv-\Av\xv\rVert_2^2 +  \sum_{g\in \cG} \lambda_g \lVert \xv_g \rVert_2
\end{equation}
where $\cG$ is a partition of the set of indices~\cite{yuan2006grouplasso}.
In this setting, only indices within a few groups are non-zero.

Conceptually, the sparse group LASSO (SGL) problem is similar to the GL in that it combines a group structure wherein entries within only a few groups are non-zero and there is additional sparsity in every active group~\cite{simon2013sparsegrouplasso}.
This double sparsity is captured through the Lagrangian formulation given by
\begin{align} \label{eq:sparse-group-lasso}
\textstyle
\min_{\xv \in \Fb^N} \lVert \yv-\Av\xv\rVert_2^2  + \sum_{g\in \cG} \lambda_g \lVert \xv_g \rVert_2 + \lambda_{\rm e} \lVert \xv \rVert_1.
\end{align}
The Lagrangian parameters $\{ \lambda_g \}$, which penalize the group norms to encourage group-level sparsity, can be group-specific or uniform across groups, depending on the underlying structure of the problem.
Crucially, the resulting proximal operator is non-separable.
This type of denoiser precludes the application of early AMP results, but non-separable functions are addressed in more recent contributions for real-valued problems~\cite{berthier2020state,chen2021asymptotic,gerbelot2023amp}.

Complex versions of the SGL problem arise in communication settings such as preamble detection in OTFS, where the set of active preambles is sparse and the channel of every device admits a sparse representation in the delay-Doppler domain~\cite{mirri2025approximate}.
Such practical linear inverse problems over the complex field motivate the study of the complex AMP framework for non-separable, complex functions.

\begin{remark} \label{remark:ComplexAMP}
Often a problem over $\Cb^n$ can be reformulated equivalently over $\Rb^{2n}$ by representing complex numbers as two-dimensional real vectors.
However, naively trying to use the standard isomorphism between $\Cb^n$ and $\Rb^{2n}$ fails for the complex AMP formulation of \eqref{eq:InverseProblem}.
The difficulty stems, partly, from the typical AMP assumptions.
Stacking the real and imaginary parts of $\xv_0$ into a single vector $[\Re \xv_0 ; \Im \xv_0] \in \Rb^{2N}$ produces an augmented measurement matrix whose rows are pairwise dependent, violating the i.i.d. requirement.
Alternatively, arranging the components as $[\Re \xv_0 , \Im \xv_0] \in \Rb^{N \times 2}$ and retaining a single real sensing matrix $\Wv \in \Rb^{n \times N}$ preserves the i.i.d. property of the entries but decouples the two columns: the observation $\Wv [\Re \xv_0, \Im \xv_0]$ cannot reproduce the cross-coupling between real and imaginary parts inherent in complex multiplication $\Av \xv_0$.
Our approach adopts the horizontal arrangement but pairs it with a richer sensing matrix $\Wv \in \Rb^{2n \times N}$, yielding an augmented observation space of $4n$ real measurements rather than $2n$.
The additional measurements restore the coupling lost by the naive formulation, but carry more information than the original complex system.
A state collapse then constrains the augmented system to match the information content of the base formulation over $\Cb$ (see~\eqref{eq:CanonicalMap1} and the discussion that follows).
\end{remark}

The nuances outlined in Remark~\ref{remark:ComplexAMP} highlight why treatments of complex AMP cannot be obtained as straightforward extensions of results for the real case.
Maleki et al.~\cite{Maleki2013CAMP} derive the Onsager correction for the complex AMP iterate via a Taylor-series expansion.
More recently, the extension of AMP to complex-valued settings with row-wise separable denoising functions was established through an argument paralleling the \emph{Householder Dice} characterization, but relying on a complex version of Stein's lemma~\cite{burak2025jointmsgdetection}.

Our strategy for studying AMP over $\Cb$ differs from previous approaches.
Rather than deriving an alternate proof from first principles, we construct an augmented system with key structural properties.
Specifically, we examine a higher-dimensional formulation over $\Rb$, which is inherently richer than the base problem over $\Cb$ and thus not equivalent to it.
Existing AMP results~\cite{bayati2011dynamics,javanmard2013state,berthier2020state,gerbelot2023amp} readily apply to the augmented system.
By collapsing this augmented system onto the complex domain through a linear transformation and restricting the class of admissible denoisers, we show that the statistical properties of the base formulation over $\Cb$ can be inferred from the richer system.
Since \emph{state evolution} (SE) depends on these statistical properties, this strategy (under suitable assumptions) yields an asymptotic characterization of complex AMP with non-separable denoisers.
The details of this approach will be developed in the following sections.

\vspace{0.5em} \noindent
\textbf{Base System over $\Cb$:}
Mathematically, we wish to establish state evolution for the composite AMP algorithm over $\Cb$:
\begin{align}
\zv^t &= \yv-\Av \xv^t + \frac{\zv^{t-1}}{\delta} \big\langle (\eta^{t-1})'(\rv^{t-1}) \big\rangle \label{eq:ResidualError} \\
\rv^t &= \xv^t + \Av^{\hermitian} \zv^t \\
\xv^{t+1} &= \eta^t(\rv^t)
\label{eq:StateEstimate}
\end{align}
where $\eta^t: \Cb^N \to \Cb^N$ is a sequence of (potentially) non-separable complex-valued vector maps.
Therein, $\zv^t \in \Cb^{n}$ is the \emph{residual error} enhanced by an Onsager term;
$\rv^t \in \Cb^{N}$ is the \emph{effective observation}; and $\xv^{t+1}$ is the \emph{state estimate} after denoising.
Matrix $\Av$ has $\cC\cN(0, 1/n)$ i.i.d. entries.

\vspace{0.5em} \noindent
\textbf{Augmented System over $\Rb$:}
Consider the matrix inverse problem over $\Rb$ defined as
\begin{equation} \label{eq:MatrixInverseProblem}
\Ov = \Wv \Sv_0 + \Nv
\end{equation}
where $\Wv \in \Rb^{2n \times N}$ is the sensing matrix, $\Sv_0 \in \Rb^{N \times 2}$ is the state matrix, and $\Nv$ denotes additive noise.
The observation matrix is then $\Ov \in \Rb^{2n \times 2}$.
Note that $\Ov$ has twice as many real dimensions as $\yv$ in \eqref{eq:InverseProblem}.
Prior work on AMP for matrix inverse problems over $\Rb$~\cite{javanmard2013state,feng2022unifying,gerbelot2023amp} provides the conceptual foundation for our analysis.

In constructing the augmented system, the relation between \eqref{eq:InverseProblem} and \eqref{eq:MatrixInverseProblem} is given by $\Sv_0 = [ \Re \xv_0 , \Im \xv_0 ]$ and
\begin{xalignat}{2} \label{eq:CanonicalMap1}
\Wv &= \begin{bmatrix} \Wv_1 \\ \vdots \\ \Wv_n \end{bmatrix} &
\Wv_{\iota}
= \begin{bmatrix} \Re \Av_{\iota,:} \\ \Im \Av_{\iota, :} \end{bmatrix} .
\end{xalignat}
The entries of $\Wv$ are i.i.d. Gaussian random variables, and the state matrix $\Sv_0$ is in one-to-one correspondence with $\xv_0$ in \eqref{eq:InverseProblem}.
The key difference lies in the observation space: above, we have $\Ov \in \Rb^{2n \times 2}$, which amounts to $4n$ real measurements, compared to the $2n$ real measurements implicit in $\yv$.
This distinction explains why we describe the augmented system as information rich and, at the same time, motivates the need to collapse it into a lower-dimensional formulation.

\vspace{0.5em} \noindent
\textbf{State Collapse:}
Let linear transformation $\Rv: \Cb^{2n} \to \Cb^{n}$ be defined through the Kronecker product $\Rv = \Iv_n \otimes \rho^\transpose$ where $\rho^\transpose = [ 1 , \pzi ]$.
Then, we can immediately write
\begin{xalignat}{3} \label{eq:StateCollapse1}
\Rv \Wv &= \Av & \Sv_0 \rho &= \xv_0 & \Rv \Nv \rho &\stackrel{d}{=} \wv
\end{xalignat}
where $\stackrel{d}{=}$ denotes equality in distribution.
That is, conditioned on $\Av$ and $\xv_0$, we have
\begin{equation*}
\Rv \Ov \rho = (\Rv \Wv) (\Sv_0 \rho) + \Rv \Nv \rho \stackrel{d}{=} \Av \xv_0 + \wv = \yv .
\end{equation*}
The connection between the two systems is manifest.
It remains to show that the statistical properties of the complex AMP iterative algorithm \eqref{eq:ResidualError}--\eqref{eq:StateEstimate} can likewise be derived from the properties of the matrix AMP iterative algorithm.
Establishing this result, under suitable restrictions, is a key contribution of this article.
The precise formulation of these restrictions, and the resulting state evolution theorem, are the subject of the next section.

\section{State Evolution for Complex AMP}
\label{sec:main-result}

Having established the relationship between the base and augmented systems, we now formalize the reduction.
Following standard practice in the AMP literature~\cite{berthier2020state}, we work with the centralized formulation and establish state evolution for the following complex AMP iteration:
\begin{align}
\uv^{k+1} &= \Av^{\hermitian} \pzg^k \left( \vv^k \right) -  c^k \pzf^k \left( \uv^k \right)
\label{eq:complex-amp-1} \\
\vv^k &= \Av \pzf^k \left( \uv^k \right) - b^k \pzg^{k-1} \left( \vv^{k-1} \right)
\label{eq:complex-amp-2}
\end{align}
where $\pzf^k:\Cb^N \to \Cb^N$ and $\pzg^k:\Cb^n \to \Cb^n$ are sequences of complex vector functions, $c^k = \langle (\pzg^{k})'(\vv^{k}) \rangle$, and $b^k = ({1}/{\delta}) \langle (\pzf^{k})'(\uv^{k})\rangle$.
In the real-valued setting, the following matrix equations have been studied~\cite{gerbelot2023amp},
\begin{xalignat}{2}
\Hv^{k+1} &= \Wv^{\transpose} \Qv^k - \Mv^k (\Cv^k)^{\transpose}
& \Mv^{k+1} &= f^{k+1} (\Hv^{k+1}) \label{eq:real-mat-1} \\
\Ev^{k} &= \Wv \Mv^{k} - \Qv^{k-1} (\Bv^{k})^{\transpose} &
\Qv^k &= g^k(\Ev^{k}) \label{eq:real-mat-2}
\end{xalignat}
where $\Wv \in \Rb^{m \times M}$ is a real-valued matrix, $f^{k+1}$ and $g^k$ are real-valued matrix functions of suitable dimensions, and
\begin{xalignat*}{2}
\Cv^k &= \frac{1}{m} \sum_{j=1}^{m} \frac{\partial (g^k_{j,:}(\Ev^k))^\transpose}{\partial \Ev_{j,:}} &
\Bv^{k} &= \frac{1}{m} \sum_{j=1}^{M} \frac{\partial (f^k_{j,:}(\Hv^k))^\transpose}{\partial \Hv_{j,:}}.
\end{xalignat*}
We next establish the reduction of \eqref{eq:real-mat-1}--\eqref{eq:real-mat-2} to \eqref{eq:complex-amp-1}--\eqref{eq:complex-amp-2} via the transformation given below.

\subsection{Canonical Transformation}
\label{sec:canonical-transform}

\begin{definition}[Canonical Transformation]
\label{def:canonical-transformation}

Recall the relations in \eqref{eq:CanonicalMap1} between $\Wv$ and $\Av$ and between $\Sv_0$ and $\xv_0$.
Set and partition the iterates of the real matrix AMP system \eqref{eq:real-mat-1}--\eqref{eq:real-mat-2} into blocks $\Hv^k \in (\Rb^{1\times 2})^N$ and $\Ev^k \in (\Rb^{2\times 2})^n$, akin to \eqref{eq:CanonicalMap1}, so that
\begin{xalignat*}{2}
\Hv^k &= \begin{bmatrix} \Hv^k_1 \\ \vdots \\ \Hv^k_N \end{bmatrix} & \Hv^k_\iota &\in \Rb^{1 \times 2}, \\
\Ev^k &= \begin{bmatrix} \Ev^k_1 \\ \vdots \\ \Ev^k_n \end{bmatrix} & \Ev^k_\iota &= \begin{bmatrix} \Ev^k_{\iota, 1, 1} & \Ev^k_{\iota, 1, 2} \\ \Ev^k_{\iota, 2, 1} & \Ev^k_{\iota, 2, 2} \end{bmatrix} \in \Rb^{2\times 2}.
\end{xalignat*}
Define the complex collapse of each block by
\begin{xalignat*}{2}
\uv^k_\iota &= \Hv^k_\iota \, \rho &
\vv^k_\iota &= \rho^\transpose \Ev^k_\iota \, \rho .
\end{xalignat*}
Given the complex functions $\pzf^k$ and $\pzg^k$ from \eqref{eq:complex-amp-1}--\eqref{eq:complex-amp-2}, define the real matrix functions $f^k$ and $g^k$ in \eqref{eq:real-mat-1}--\eqref{eq:real-mat-2} so that $\Mv^k = f^k(\Hv^k)$ with $\Mv^k_\iota = \big[ \Re \pzf^k_\iota(\uv^k), \Im \pzf^k_\iota(\uv^k) \big]$, and that $\Qv^k = g^k(\Ev^k)$ has the block-wise form
\begin{xalignat*}{2}
&\Qv^k = \begin{bmatrix} \Qv^k_1 \\ \vdots \\ \Qv^k_n \end{bmatrix} &
\Qv^k_\iota &= \begin{bmatrix}
\Re \pzg^k_\iota(\vv^k) & \Im \pzg^k_\iota(\vv^k) \\
\Im \pzg^k_\iota(\vv^k) & -\Re \pzg^k_\iota(\vv^k)
\end{bmatrix} .
\end{xalignat*}
\end{definition}

\begin{remark}
The $2\times 2$ block form of $\Qv^k_\iota$ encodes the action of complex multiplication by $\overline{\rho}$: a matrix of the form $\begin{bsmallmatrix} a & b \\ b & -a \end{bsmallmatrix}$ satisfies
$\begin{bsmallmatrix} a & b \\ b & -a \end{bsmallmatrix}\rho = (a + \pzi b)\overline{\rho}$,
which is precisely Lemma~\ref{lem:prop}\ref{lem:qri} below.
This structure is a key design choice of the canonical transformation.
\end{remark}

\begin{remark}[Admissible Functions]
Definition~\ref{def:canonical-transformation} restricts the class of real-valued matrix functions that are compatible with the collapse.
A generic $g^k$ would produce blocks $\Qv^k_\iota \in \Rb^{2\times 2}$ with four degrees of freedom; the imposed structure $\begin{bsmallmatrix} a & b \\ b & -a \end{bsmallmatrix}$ has only two degrees of freedom, forcing $g^k$ to depend on~$\Ev^k_\iota$ solely through the collapsed variable $\vv^k_\iota = \rho^\transpose \Ev^k_\iota \rho$.
Similarly, for each component $\iota \in [N]$, $f^k_\iota$ is constrained to be the real--imaginary encoding of $\pzf^k_\iota$ acting on $\Hv^k \rho$.
In short, complex-compatible nonlinearities, those whose real matrix representation is fully determined by their complex counterpart, admit the reduction.
\end{remark}

\begin{remark}[Collapsed System]
The transformation is many-to-one: a block $\Ev_{\iota}^k \in \Rb^{2 \times 2}$ is mapped to a scalar $\vv_\iota \in \Cb \cong \Rb^2$, so information is lost.
Because $f^k$ and $g^k$ are restricted to the admissible class above, the orbit of the resulting system can be tracked entirely in the lower-dimensional complex domain.
Importantly, all components required for evaluating \eqref{eq:complex-amp-1}--\eqref{eq:complex-amp-2} are preserved under this mapping.
\end{remark}

Properties of the canonical transformation and its effect on the AMP dynamics are established below.

\begin{lemma}[Properties of Canonical Transformation] \label{lem:prop}
For the mapping of Definition~\ref{def:canonical-transformation}, the following equalities hold:
\begin{enumerate}[label=(\roman*), leftmargin=3em, labelwidth=3.5em, labelsep=0.5em]
\item $\Rv \Wv = \Av$, $\Hv^k \rho = \uv^k$, $\Rv \Ev^k \rho = \vv^k$, $\Mv^k \rho = \pzf^k(\uv^k)$.\label{lem:eqs}
\item $\Qv_\iota^k \rho = \pzg_\iota^k(\vv^k) \overline{\rho} \;$ for all $\iota \in [n]$.\label{lem:qri}
\item $\Qv^k \rho = \Rv^\hermitian \pzg^k(\vv^k)$, $\Rv \Qv^k = \pzg^k(\vv^k) \rho^\hermitian$.\label{lem:qr} \label{lem:rq}
\item $\frac{1}{2} \rho^\hermitian \frac{\partial}{\partial (\Hv_{\iota, :}^k)^\transpose} = \frac{\partial}{\partial \uv_\iota} \;$ for all $\iota \in [N]$.\label{lem:hu}
\item $\frac{1}{2} \begin{bmatrix}
    \frac{\partial}{\partial (\Ev^k_{\iota, 1, :})^\transpose} &
    \frac{\partial}{\partial (\Ev^k_{\iota, 2, :})^\transpose}
\end{bmatrix} \overline{\rho} = \rho \frac{\partial}{\partial \vv_\iota^k} \;$ for all $\iota\in [n]$.\label{lem:ev}
\end{enumerate}
\end{lemma}
\begin{IEEEproof}
Throughout, recall the following notations and properties,
\begin{xalignat*}{3}
\rho &= \begin{bmatrix} 1 \\ \pzi \end{bmatrix} &
\overline{\rho} &= \begin{bmatrix} 1 \\ -\pzi \end{bmatrix} &
\rho^\hermitian \rho &= 2
\end{xalignat*}
and $\Rv = \Iv_n \otimes \rho^\transpose$.

Consider \ref{lem:eqs}.
The identity $\Rv \Wv = \Av$ was introduced in \eqref{eq:StateCollapse1}.
For each $\iota \in [n]$, the $\iota$-th block of $\Rv \Wv$ is
\begin{equation*}
\rho^\transpose \Wv_\iota
= \begin{bmatrix} 1 & \pzi \end{bmatrix} \begin{bmatrix} \Re \Av_{\iota,:} \\ \Im \Av_{\iota,:} \end{bmatrix}
= \Re \Av_{\iota,:} + \pzi \, \Im \Av_{\iota,:} = \Av_{\iota,:} .
\end{equation*}
The identity $\Hv^k \rho = \uv^k$ is immediate from Definition~\ref{def:canonical-transformation}, which states that $\uv^k_\iota = \Hv^k_\iota \rho$.
For $\Rv \Ev^k \rho = \vv^k$, note that the $\iota$-th component of $\Rv \Ev^k \rho$ is $\rho^\transpose \Ev^k_\iota \rho = \vv^k_\iota$, again by definition.
To verify the expansion explicitly,
\begin{equation*}
\begin{split}
\vv^k_\iota &= \left[ \Rv \Ev^k \rho \right]_{\iota}
= \left[ \left( \Iv_n \otimes \rho^\transpose \right) \Ev^k \rho \right]_{\iota} \\
&= \rho^\transpose \Ev^k_\iota \rho = \begin{bmatrix} 1 & \pzi \end{bmatrix} \begin{bmatrix} \Ev^k_{\iota,1,1} & \Ev^k_{\iota,1,2} \\ \Ev^k_{\iota,2,1} & \Ev^k_{\iota,2,2} \end{bmatrix} \begin{bmatrix} 1 \\ \pzi \end{bmatrix} \\
&= (\Ev^k_{\iota,1,1} - \Ev^k_{\iota,2,2}) + \pzi (\Ev^k_{\iota,1,2} + \Ev^k_{\iota,2,1}) .
\end{split}
\end{equation*}
The identity $\Mv^k \rho = \pzf^k(\uv^k)$ follows by the same reasoning as $\Hv^k \rho = \uv^k$, since by definition $\Mv^k_\iota = \big[ \Re \pzf^k_\iota(\uv^k), \Im \pzf^k_\iota(\uv^k) \big]$.

For \ref{lem:qri}, write $a_\iota = \Re \pzg^k_\iota(\vv^k)$ and $b_\iota = \Im \pzg^k_\iota(\vv^k)$.
By Definition~\ref{def:canonical-transformation}, we have
\begin{equation*}
\begin{split}
\Qv^k_\iota \rho &= \begin{bmatrix} a_\iota & b_\iota \\ b_\iota & -a_\iota \end{bmatrix} \begin{bmatrix} 1 \\ \pzi \end{bmatrix} = \begin{bmatrix} a_\iota + \pzi b_\iota \\ b_\iota - \pzi a_\iota \end{bmatrix} \\
 &= (a_\iota + \pzi b_\iota) \begin{bmatrix} 1 \\ -\pzi \end{bmatrix} = \pzg^k_\iota(\vv^k) \, \overline{\rho} .
\end{split}
\end{equation*}

Consider \ref{lem:qr}.
Stacking \ref{lem:qri} over all $\iota \in [n]$ and noting that $\Rv^\hermitian = \Iv_n \otimes \overline{\rho}$, the $\iota$-th block of $\Rv^\hermitian \pzg^k(\vv^k)$ is $\overline{\rho} \, \pzg^k_\iota(\vv^k) = \pzg^k_\iota(\vv^k) \overline{\rho}$, establishing $\Qv^k \rho = \Rv^\hermitian \pzg^k(\vv^k)$.
For the second identity, the $\iota$-th component of $\Rv \Qv^k$ is
\begin{equation*}
\begin{split}
\rho^\transpose \Qv^k_\iota &= \begin{bmatrix} 1 & \pzi \end{bmatrix} \begin{bmatrix} a_\iota & b_\iota \\ b_\iota & -a_\iota \end{bmatrix} \\
&= (a_\iota + \pzi b_\iota) \begin{bmatrix} 1 & -\pzi \end{bmatrix} = \pzg^k_\iota(\vv^k) \, \rho^\hermitian .
\end{split}
\end{equation*}

To establish \ref{lem:hu}, we write $\Hv^k_{\iota, :} = [ h_1, h_2 ]$ with $h_1 = \Re \uv^k_\iota$ and $h_2 = \Im \uv^k_\iota$.
Then
\begin{equation*}
\frac{1}{2} \rho^\hermitian \frac{\partial}{\partial (\Hv^k_{\iota,:})^\transpose} = \frac{1}{2} \begin{bmatrix} 1 & -\pzi \end{bmatrix} \begin{bmatrix} \partial / \partial h_1 \\ \partial / \partial h_2 \end{bmatrix} = \frac{1}{2} \left( \frac{\partial}{\partial h_1} - \pzi \frac{\partial}{\partial h_2} \right) .
\end{equation*}
This is precisely the Wirtinger derivative $\frac{\partial}{\partial \uv_\iota}$ with respect to $\uv_\iota = h_1 + \pzi h_2$.

Finally, \ref{lem:ev} can be obtained by writing the entries of $\Ev^k_{\iota, :, :} \in \Rb^{2 \times 2}$ as $e_{jl}$ for $j, l \in \{1, 2\}$.
The left-hand side is
\begin{equation*}
\frac{1}{2} \begin{bmatrix} \partial/\partial e_{11} & \partial/\partial e_{21} \\ \partial/\partial e_{12} & \partial/\partial e_{22} \end{bmatrix} \overline{\rho} = \frac{1}{2} \begin{bmatrix} \partial/\partial e_{11} - \pzi \, \partial/\partial e_{21} \\ \partial/\partial e_{12} - \pzi \, \partial/\partial e_{22} \end{bmatrix} .
\end{equation*}
Since $g^k$ depends on $\Ev^k_\iota$ only through $\vv^k_\iota = (e_{11} - e_{22}) + \pzi(e_{12} + e_{21})$, the chain rule gives
$\frac{\partial}{\partial e_{11}} = \frac{\partial}{\partial (\Re \vv_\iota)}$, $\frac{\partial}{\partial e_{22}} = -\frac{\partial}{\partial (\Re \vv_\iota)}$, $\frac{\partial}{\partial e_{12}} = \frac{\partial}{\partial (\Im \vv_\iota)}$, and $\frac{\partial}{\partial e_{21}} = \frac{\partial}{\partial (\Im \vv_\iota)}$.
Substituting,
\begin{equation*}
\begin{split}
&\frac{1}{2} \begin{bmatrix} \partial/\partial (\Re \vv_\iota) - \pzi \, \partial/\partial (\Im \vv_\iota) \\ \partial/\partial (\Im \vv_\iota) + \pzi \, \partial/\partial (\Re \vv_\iota) \end{bmatrix}
 \\
&\quad= \frac{1}{2} \left( \frac{\partial}{\partial (\Re \vv_\iota)} - \pzi \frac{\partial}{\partial (\Im \vv_\iota)} \right) \begin{bmatrix} 1 \\ \pzi \end{bmatrix}
= \frac{\partial}{\partial \vv^k_\iota} \, \rho,
\end{split}
\end{equation*}
where the last step again uses the Wirtinger derivative $\frac{\partial}{\partial \vv_\iota} = \frac{1}{2}\!\left( \frac{\partial}{\partial (\Re \vv_\iota)} - \pzi \frac{\partial}{\partial (\Im \vv_\iota)} \right)$ and the factorization of the column vector as a scalar times $\rho$.
\end{IEEEproof}

\begin{proposition}[Dynamics of Collapsed System]
\label{prop:red-dynamics}
     Under the canonical transformation of Definition~\ref{def:canonical-transformation}, with $f^k$ and $g^k$ restricted to the admissible class therein, \eqref{eq:real-mat-1}--\eqref{eq:real-mat-2} imply \eqref{eq:complex-amp-1}--\eqref{eq:complex-amp-2}.
\end{proposition}
\begin{IEEEproof} Using Lemma~\ref{lem:prop}\ref{lem:eqs}, we can write
\begin{equation*}
\uv^{k+1} = \Hv^{k+1} \rho = \Wv^\transpose \Qv^k \rho - \Mv^k (\Cv^k)^\transpose \rho.
\end{equation*}
We can simplify the first term using Lemma~\ref{lem:prop}\ref{lem:qr} as
\begin{equation*}
\Wv^\transpose \Qv^k \rho = \Wv^\transpose \Rv^\hermitian \pzg^k(\vv^k)
= (\Rv \Wv)^\hermitian \pzg^k(\vv^k) = \Av^\hermitian \pzg^k(\vv^k).
\end{equation*}
For the second term, we need to compute
\begin{equation*}
\begin{split}
(\Cv^k)^\transpose
&= \textstyle \frac{1}{n} \sum_{\iota=1}^n \frac{1}{2} \begin{bmatrix}
\frac{\partial}{\partial (\Ev^k_{\iota, 1, :})^\transpose} &
\frac{\partial}{\partial (\Ev^k_{\iota, 2, :})^\transpose}
\end{bmatrix} \Qv_\iota^k.
\end{split}
\end{equation*}
Using Lemma~\ref{lem:prop}\ref{lem:qri} and Lemma~\ref{lem:prop}\ref{lem:ev}, we get $(\Cv^k)^\transpose \rho = \langle (\pzg^k)'(\vv^k) \rangle \rho$, then $\Mv^k (\Cv^k)^\transpose \rho = \langle (\pzg^k)'(\vv^k) \rangle \pzf^k(\uv^k)$.
This proves \eqref{eq:complex-amp-1}.

We turn to the second equation in the composite iteration,
\begin{equation*}
\vv^{k} = \Rv \Ev^k \rho = \Rv \Wv \Mv^k \rho - \Rv \Qv^{k-1} (\Bv^k)^\transpose \rho.
\end{equation*}
Using Lemma~\ref{lem:prop}\ref{lem:eqs}, we readily obtain $\Rv \Wv \Mv^k \rho = \Av \pzf^k(\uv^k)$.

Proceeding to the second term, we expand $(\Bv^{k})^\transpose = \frac{1}{n}\sum_{\iota=1}^N \left(\frac{1}{2} \frac{\partial}{\partial (\Hv_{\iota,:}^k)^\transpose} \right) (\Mv^k_{\iota, :})$.
Applying Lemma~\ref{lem:prop}\ref{lem:rq} and Lemma~\ref{lem:prop}\ref{lem:hu}, we get
\begin{equation*}
\begin{split}
\Rv \Qv^{k-1} (\Bv^k)^\transpose \rho
&= \frac{1}{\delta} \langle (\pzf^k)'(\uv^k)\rangle \pzg^{k-1}(\vv^{k-1}),
\end{split}
\end{equation*}
which proves \eqref{eq:complex-amp-2}.
\end{IEEEproof}

Thus, the dynamics of the augmented real system, when collapsed through the canonical transformation, reproduce the complex AMP iteration \eqref{eq:complex-amp-1}--\eqref{eq:complex-amp-2}.
It remains to establish state evolution for this system.

\subsection{Regularity Assumptions and State Evolution}
\label{sec:state-evolution}

The proof strategy proceeds in three stages.
First, we state regularity conditions on the complex functions $\pzf^k$ and $\pzg^k$ (Assumption~\ref{assumption:reg-complex} below).
These conditions are deliberately chosen so that, when combined with the canonical transformation of Definition~\ref{def:canonical-transformation}, the induced real matrix functions $f^k$ and $g^k$ satisfy the regularity assumptions of the real matrix AMP framework in~\cite{gerbelot2023amp}.
Since $f^k$ and $g^k$ are defined through $\pzf^k$ and $\pzg^k$ via the admissible class of Definition~\ref{def:canonical-transformation}, this transfer of regularity is natural, though it must still be verified.
Second, having established that the augmented real system meets the conditions of~\cite{gerbelot2023amp}, we invoke the state evolution result therein.
Third, we collapse back through the canonical transformation to obtain state evolution for the complex system (Theorem~\ref{thm:complex-amp-main}).

\begin{assumption}[Regularity Conditions for Complex AMP]
\label{assumption:reg-complex}
Consider a sequence of recursions, each of the form \eqref{eq:complex-amp-1}--\eqref{eq:complex-amp-2}, indexed by $n$ and $N = N(n)$ with $n/N \to \delta \in (0, \infty)$ as $n\to \infty$.
In this regime, assume the following holds.
\begin{list}{(A\arabic{enumi})}
  {\usecounter{enumi}
  \setlength{\leftmargin}{3em}
  \setlength{\labelwidth}{3.5em}
  \setlength{\labelsep}{0.5em}}
\item $\Av_{i, j} \equiv \Av_{i, j}(n) \sim \cC\cN(0, 1/n)$ i.i.d.
\item The sequences of functions $\{\pzf^k\}$ and $\{\pzg^k\}$ are uniformly Lipschitz, where uniformity is with respect to $n$.
\item For $\uv^0 \equiv \uv^0(n)$, $\norm{\uv^0}_2 / \sqrt{n}$ converges as $n\to \infty$.
\item The following limits exist and are finite:
\begin{xalignat*}{2}
\lim_{n \uparrow \infty} & \tfrac{1}{n} \langle \pzf^0(\uv^0), \pzf^0(\uv^0) \rangle, &
\lim_{n \uparrow \infty} & \tfrac{1}{n} \langle \pzf^0(\uv^0), \overline{\pzf^0(\uv^0)} \rangle.
\end{xalignat*}
\item For any $\sigma > 0$ and $c \in \Cb$, the following limits exist for $\Zv \sim \cC\cN(\zerov, \sigma^2 \otimes \Iv_N, c \otimes \Iv_N)$ and are finite:
\begin{xalignat*}{2}
\lim_{n \uparrow \infty} & \tfrac{1}{n} \Eb[ \langle \pzf^0(\uv^0), \pzf^0(\Zv) \rangle], &
\lim_{n \uparrow \infty} & \tfrac{1}{n} \Eb[ \langle \pzf^0(\uv^0), \overline{\pzf^0(\Zv)} \rangle].
\end{xalignat*}
\item For any $r, s > 0$, $\Sv \in \Cb^{2\times 2}$ Hermitian positive definite matrix, and $\Cv \in \Cb^{2\times 2}$ complex symmetric matrix, the following limits exist and are finite:
\begin{align*}
\lim_{n\uparrow \infty} & \tfrac{1}{n} \Eb[\langle \pzf^r(\Zv_\tau^r), \pzf^s(\Zv_\tau^s) \rangle],
\lim_{n\uparrow \infty} \tfrac{1}{n} \Eb[\langle \pzf^r(\Zv_\tau^r), \overline{\pzf^s(\Zv_\tau^s)} \rangle] \\
\lim_{n\uparrow \infty} & \tfrac{1}{n} \Eb[\langle \pzg^r(\Zv_\sigma^r), \pzg^s(\Zv_\sigma^s) \rangle],
\lim_{n\uparrow \infty} \tfrac{1}{n} \Eb[\langle \pzg^r(\Zv_\sigma^r), \overline{\pzg^s(\Zv_\sigma^s)} \rangle]
\end{align*}
where
    $(\Zv_\tau^r, \Zv_\tau^s) \sim \cC\cN(\zerov, \Sv \otimes \Iv_N, \Cv \otimes \Iv_N)$ and $(\Zv_\sigma^r, \Zv_\sigma^s) \sim \cC\cN(\zerov, \Sv \otimes \Iv_n, \Cv \otimes \Iv_n)$.
\end{list}
\end{assumption}

The reader familiar with the AMP literature will recognize that Assumption~\ref{assumption:reg-complex} enables the definition of two doubly infinite arrays of complex scalars, $(\sigma^{s, r}, \tau^{s,r})_{s,r\le t}$, through the usual recursion stated below.

\begin{definition}[Complex State Evolution]
\label{def:state-evol-complex}
Suppose Assumption~\ref{assumption:reg-complex} holds. Define $(\sigma^{s, r}, \tau^{s,r})_{s,r\le t}$ through the following recursion. Given $\uv^0\in \Cb^N$, define the initialization
\begin{align}
\label{eq:init-complex}
    \sigma^{0,0} \triangleq \frac{1}{\delta} \lim_{N \to \infty}  \frac{1}{N} \langle \pzf^0(\uv^0), \pzf^0(\uv^0) \rangle \in \Cb.
\end{align}
Assuming that $(\sigma_{s,r})_{0\le s, r\le t}, (\tau_{s,r})_{1\le s,r\le t}$ are available, denoting $(\Zv_{\sigma}^0, \dots, \Zv_{\sigma}^t) \sim \cC\cN(\zerov, (\sigma_{s,r})_{s,r\le t} \otimes \Iv_n)$, define for $0\le s \le t$:
\begin{align}
\label{eq:tau-complex}
    \overline{\tau^{t+1, s+1}} = \tau^{s+1, t+1} =  \lim_{n \to \infty} \frac{1}{n} \Eb[\langle \pzg^s(\Zv_\sigma^s), \pzg^t(\Zv_\sigma^t) \rangle],
\end{align}
which characterizes $(\tau_{s,r})_{1\le s,r\le t+1}$. Next, set $(\Zv_\tau^1, \dots, \Zv_{\tau}^{t+1}) \sim \cC\cN(\zerov, (\tau_{s,r})_{1\le s,r \le t+1} \otimes \Iv_N)$, and define
\begin{align}
\label{eq:sigma-complex}
    \overline{\sigma^{t+1, s}} = \sigma^{s,t+1} = \frac{1}{\delta} \lim_{N\to \infty} \frac{1}{N} \Eb[\langle \pzf^s(\Zv_\tau^s), \pzf^{t+1}(\Zv_\tau^{t+1})\rangle],
\end{align}
which characterizes $(\sigma_{s,r})_{0\le s,r \le t+1}$.
\end{definition}

\begin{remark}
In the real matrix AMP setting~\cite{gerbelot2023amp}, the analogous SE arrays $(\Sigmav^{s,r}, \Tv^{s,r})$ are $2\times 2$ real matrices (see Definition~\ref{def:state-evol-real} in the appendix).
By contrast, the complex SE arrays $(\sigma^{s,r}, \tau^{s,r})$ are complex scalars.
This reduction from $2\times 2$ matrices to scalars is a direct consequence of the canonical transformation and reflects the rotational symmetry of the complex Gaussian distribution.
\end{remark}

With the SE arrays in hand, we can now state the main result.

\begin{theorem}[State Evolution for Complex AMP]
\label{thm:complex-amp-main}
Consider the complex AMP system~\eqref{eq:complex-amp-1}--\eqref{eq:complex-amp-2} and suppose Assumption~\ref{assumption:reg-complex} holds.
Let $\left( \Zv_{\sigma}^0, \ldots, \Zv_{\sigma}^t \right)$ and $\left( \Zv_{\tau}^1, \ldots, \Zv_{\tau}^t \right)$ be centered complex Gaussian vectors of covariances $\left( \sigma_{s, r} \right)_{0\le s, r \leq t} \otimes \Iv_n$ and  $\left( \tau_{s, r} \right)_{1\le s, r \leq t} \otimes \Iv_N$, respectively.
Then, for any sequence of pseudo-Lipschitz functions $\Phiv_n : \left( \mathbb{C}^N \times \mathbb{C}^n \right)^t \times \mathbb{C}^N \rightarrow \mathbb{C}$,
\begin{equation*}
\begin{split}
&\Phiv_n \left( \uv^0, \vv^0, \ldots, \uv^{t-1}, \vv^{t-1}, \uv^t \right) \\
&\hspace{20pt} \stackrel{\mathrm{P}}{\simeq} \mathbb{E} \left[ \Phiv_n \left( \uv^0, \Zv_{\sigma}^0, \ldots, \Zv_{\tau}^{t-1}, \Zv_{\sigma}^{t-1}, \Zv_{\tau}^t \right) \right].
\end{split}
\end{equation*}
\end{theorem}
\begin{IEEEproof} This theorem establishes conditions under which the AMP dynamics defined in \eqref{eq:complex-amp-1} and \eqref{eq:complex-amp-2} satisfy the requirements for state evolution in Definition~\ref{def:state-evol-complex}.
Specifically, if the complex AMP system satisfies Assumption~\ref{assumption:reg-complex}, then, under the canonical transformation of Definition~\ref{def:canonical-transformation}, the augmented real-valued matrix AMP system satisfies the conditions required to define state evolution for the real system (see Appendix~\ref{sec:regularity-transfer}).

The distributions of $\Zv_\sigma$ and $\Zv_\tau$ in the complex state evolution can be derived, via straightforward computations, from the corresponding distributions in the state evolution equations of the real system (see Appendix~\ref{sec:state-evolution-collapse}), with existence of the necessary terms being guaranteed by Assumption~\ref{assumption:reg-complex}.

Conceptually, the lifting of the complex-valued pseudo-Lipschitz functions to the augmented real scenario induces functions that are structured real-valued pseudo-Lipschitz functions (see Appendix~\ref{sec:pseudo-lipschitz-functions}).
Collectively, they are included in the class of functions for which \emph{state evolution} (SE) applies, per \cite[Theorem~2]{gerbelot2023amp}.
In turn, this mapping implicitly dictates the properties of the original system, after collapse.

By Proposition~\ref{prop:red-dynamics}, the orbit of the collapsed system aligns with that of the complex AMP system. Since the associated state evolution distributions are consistent under this correspondence, it follows that the asymptotic statistical properties of the complex AMP system are inherited from those of the reduced real system. A detailed proof can be found in Appendices~\ref{sec:regularity-transfer}, \ref{sec:state-evolution-collapse}, and \ref{sec:pseudo-lipschitz-functions}.

\end{IEEEproof}

When the functions $\eta^t$ in \eqref{eq:StateEstimate} satisfy the regularity conditions similar to $\pzf^t$ in Assumption~\ref{assumption:reg-complex}, we obtain the following corollary using standard transformations in the AMP literature \cite[Theorem~14]{berthier2020state}, but in the complex domain.
This characterizes the mean-squared error (MSE) in the estimates of the AMP algorithm, \eqref{eq:ResidualError}--\eqref{eq:StateEstimate}.

\begin{corollary}[Complex AMP for Vector Estimation]
\label{cor:complex-amp-vec-est}
Let $\yv \in \Cb^n, \Av \in \Cb^{n\times N}$ be as given, let $\xv_0$ be the true vector, and let $\sigma_w^2$ denote the noise variance. Starting from an initial guess $\xv^0$, the mean squared error in the estimate $\xv^t$, following \eqref{eq:ResidualError}, \eqref{eq:StateEstimate} at time $t$ satisfies
\begin{equation*}
\textstyle
\lim_{N\to \infty} \frac{1}{N} \lVert \xv^{t} - \xv_0 \rVert_2^2 = \delta(\tau_{t+1}^2-\sigma_w^2),
\end{equation*}
where $\tau_{t+1}^2 \equiv \tau^{t+1,t+1}$ denotes the diagonal entry of the state evolution array in Definition~\ref{def:state-evol-complex}, and for $\Zv \sim \cC\cN(0, \Iv_N)$, we have the recursion:
\begin{equation*}
\textstyle
\tau_{t+1}^2 = \sigma_w^2 + \frac{1}{\delta} \lim_{N\to \infty} \frac{1}{N} \Eb[\norm{\eta^t(\xv_0 + \tau_t \Zv)-\xv_0}_2^2],
\end{equation*}
where $\tau_0^2 = \sigma_w^2 + \frac{1}{\delta} \lim_{N\to \infty} \frac{1}{N} \Eb[\norm{\xv_0}_2^2]$.
\end{corollary}
\begin{IEEEproof}
The proof consists in constructing the functions $\pzf^t, \pzg^t$ explicitly in terms of $\eta^t$ and the associated vectors; details are given in Appendix~\ref{sec:amp-algo-corollary}.
\end{IEEEproof}

In summary, focusing on complex AMP with non-separable functions, we propose a method to construct a richer real-valued system with desirable properties.
This augmented system is well suited for applying existing AMP results~\cite[Theorem~2]{gerbelot2023amp}.
Still, this higher-dimensional variant is not equivalent to the original complex AMP setting, as the resulting space inherently contains more information.
To reconcile this, we introduce a many-to-one transformation that maps the augmented system to the original complex formulation while preserving its key dynamics.
Finally, we show that the statistical properties of the induced system coincide with those of the original, thereby establishing state evolution in the complex setting.

\subsection{Matrix Extension}
\label{sec:matrix-extension}

The canonical transformation framework extends from complex vectors to complex matrices.
In the matrix setting, $p$ feature vectors are estimated simultaneously from the observation model $\Yv = \Av \Xv + \Nv$, where $\Xv \in \Cb^{N\times p}$ and $\Yv \in \Cb^{n\times p}$.
The measurement matrix $\Av \in \Cb^{n\times N}$ has i.i.d. entries $\Av_{\iota,j} \sim \cC\cN(0, 1/n)$, and the noise matrix $\Nv \in \Cb^{n\times p}$ has i.i.d. entries $\Nv_{\iota,j} \sim \cC\cN(0, \sigma_w^2)$.
The Onsager corrections generalize from scalars to $p \times p$ Jacobian matrices, and the full derivation is deferred to Appendix~\ref{sec:mat-amp-reduction}.

Generalizing the collapse operator $\Rv = \Iv_n \otimes \rho^\transpose$ from Section~\ref{sec:canonical-transform}, define the matrix collapse operators
\begin{equation}
    \Rv_R = \Iv_n \otimes \rho^\transpose \in \Cb^{n\times 2n}, \qquad \Rv_L = \Iv_p \otimes \rho \in \Cb^{2p \times p}.
\end{equation}
Let $\pzf^k: \Cb^{N\times p} \to \Cb^{N\times p}$ and $\pzg^k: \Cb^{n\times p} \to \Cb^{n \times p}$ be sequences of Lipschitz functions.
The complex matrix AMP iteration is defined by
\begin{xalignat}{2}
\Uv^{k+1} \! &= \! \Av^{\hermitian} \Gv^k \! - \Fv^k (\Kv^k)^{\transpose} \!
& \Fv^{k+1} \! &= \! \pzf^{k+1} (\Uv^{k+1}) \label{eq:complex-mat-amp-U} \\
\Vv^{k} &= \Av \Fv^{k} - \Gv^{k-1} (\Lv^{k})^{\transpose} &
\Gv^k &= \! \pzg^k(\Vv^{k}) \label{eq:complex-mat-amp-V}
\end{xalignat}
with initialization $\Fv^0 = \pzf^0(\Uv^0)$ for a given $\Uv^0 \in \Cb^{N \times p}$ and $\Gv^{-1} = \mathbf{0}$, where the Jacobian matrices
\begin{align}
    (\Kv^k)^\transpose &= \frac{1}{n} \sum_{d=1}^{n} \frac{\partial \pzg^k_{d,:}(\Vv^k)}{\partial \Vv_{d,:}^\transpose} \in \Cb^{p\times p}, \\
    (\Lv^{k})^\transpose &= \frac{1}{n} \sum_{d=1}^{N} \frac{\partial \pzf^k_{d,:}(\Uv^k)}{\partial \Uv_{d,:}^\transpose} \in \Cb^{p\times p}
\end{align}
have $(\iota,j)$th entry $\left(\frac{\partial \pzg_{d,:}^k(\Vv^k)}{\partial \Vv_{d,:}^\transpose}\right)_{\iota,j} = \frac{\partial \pzg_{d,\iota}^k(\Vv^k)}{\partial \Vv_{d,j}}$, which is the Wirtinger derivative of $\pzg_{d,\iota}^k(\Vv^k)$ with respect to $\Vv_{d,j}^k$.

\begin{proposition}[Dynamics of Collapsed Matrix System]
\label{prop:mat-red-dynamics}
Under the matrix canonical transformation defined in Appendix~\ref{sec:mat-amp-denoisers}, the real matrix AMP equations~\eqref{eq:real-mat-1}--\eqref{eq:real-mat-2} imply the complex matrix AMP equations~\eqref{eq:complex-mat-amp-U}--\eqref{eq:complex-mat-amp-V}.
\end{proposition}
\begin{IEEEproof}
See Appendix~\ref{sec:mat-amp-reduction}.
\end{IEEEproof}

\begin{remark}[Extension of Regularity Conditions]
The regularity conditions of Assumption~\ref{assumption:reg-complex} extend to the matrix setting by requiring the matrix-valued functions $\pzf^k: \Cb^{N \times p} \to \Cb^{N \times p}$ and $\pzg^k: \Cb^{n \times p} \to \Cb^{n \times p}$ to satisfy uniform Lipschitz conditions with respect to the Frobenius norm.
The state evolution arrays become $p \times p$ complex matrices, and the reduction argument of Theorem~\ref{thm:complex-amp-main} carries over with parameters $m = 2n$, $M = N$, and block size~$2p$.
\end{remark}

\section{Complex Sparse Group LASSO}
\label{sec:complex_sparse_group_lasso}

We apply the developed framework to the complex sparse group LASSO, a formulation whose real-valued counterpart is well established in the statistics and signal processing literature (cf.\ Section~\ref{sec:motivation}).
In the complex domain, this problem arises naturally in wireless systems.
A primary example is OTFS preamble detection, where the set of active users is sparse and each user's channel admits a sparse representation in the delay-Doppler domain.
Importantly, the SGL proximal operator is inherently non-separable: the group-level shrinkage depends on the joint norm of all entries that survive element-wise thresholding within each group.
The non-separable state evolution established in Section~\ref{sec:main-result} is therefore not merely applicable but required for a rigorous asymptotic characterization.

We now describe the signal model and recovery formulation precisely.
The goal is to recover a structured signal $\xv \in \Cb^N$ from noisy linear measurements $\yv = \Av \xv + \wv$, where $\Av \in \Cb^{n \times N}$ has i.i.d. entries $\Av_{i,j} \sim \cC\cN(0, 1/n)$ and $\wv \sim \cC\cN(\zerov, \sigma_w^2 \Iv_n)$.
Recall that $\delta = n/N$.
We assume that the signal $\xv$ exhibits \emph{structured} group sparsity.
Specifically, there are $G$ groups, each of size $S = N / G$.
Let $k$ denote the total number of active elements, so that the sparsity ratio is $\varrho = k/n$, and let $r_{\mathrm{ga}}$ denote the fraction of active groups.
The fraction of active elements is $r_{\mathrm{a}} = k/N = \varrho\delta$, and the fraction of active elements per active group is $r_{\mathrm{ea}} = r_{\mathrm{a}} / r_{\mathrm{ga}}$.
The constraint $r_{\mathrm{ea}} \le 1$ requires $r_{\mathrm{ga}} \ge \varrho\delta$.
In our experiments, we pick $r_{\mathrm{ga}}$ fraction of $G$ groups at random to be active, and within an active group, we pick $r_{\mathrm{ea}}$ fraction of the elements at random to be active.
Active elements are drawn uniformly on the complex unit circle, i.e., $e^{\pzi\theta}$ with $\theta \sim \mathrm{Unif}[0, 2\pi)$. Additionally, we assume that the group structure is perfectly known, i.e., we know the partition of the set of indices into groups.

Owing to this double sparsity structure, a natural formulation for the recovery problem is the complex sparse group LASSO,
\begin{align}
    \label{eq:complex-sparse-group-lasso}
    \min_{\xv} ~ \norm{\yv-\Av\xv}_2^2 + \lambda_{\mathrm{g}} \sum_{g=1}^G \norm{\xv_{g,:}}_2 + \lambda_{\mathrm{e}} \norm{\xv}_1.
\end{align}
Since this optimization problem is convex but non-smooth, it can be solved via proximal gradient descent, e.g., iterative soft thresholding algorithm (ISTA).
The associated proximal operator is
\begin{equation}
    \begin{split}
    \eta(\rv, \lambda_{\mathrm{g}}, \lambda_{\mathrm{e}}) &= \operatorname*{arg\,min}_{\tilde{\xv}} \Big\{ \norm{\rv - \tilde{\xv}}_2^2 \\
    & \hspace{10pt} + \lambda_{\mathrm{g}} \sum_{g=1}^G \norm{\tilde{\xv}_{g,:}}_2 + \lambda_{\mathrm{e}} \norm{\tilde{\xv}}_1 \Big\}.
\end{split}
\end{equation}
The $j$th component of this operator admits the closed form given by
\begin{align}
    \eta_j(\rv, \lambda_{\mathrm{g}}, \lambda_{\mathrm{e}}) =
         \left( 1- \frac{\lambda_{\mathrm{g}}}{\norm{\eta^{\mathrm{st}}(\rv_{\cG_j}, \lambda_{\mathrm{e}})}_2}  \right)_+ \eta^{\mathrm{st}}(r_j, \lambda_{\mathrm{e}}),
\end{align}
where $(\cdot)_+ \triangleq\max \{0, \cdot\}$ and $\cG_j \subset [N]$ denotes the index set of the group containing element $j$.
For any positive integer $m$, the function $\eta^{\mathrm{st}}: \Cb^m \times \Rb_+ \to \Cb^m$ acts component-wise as complex soft-thresholding, with $j$th entry
\begin{align}
    \eta^{\mathrm{st}}_j(\rv, \lambda_{\mathrm{e}}) = \left(1-\frac{\lambda_{\mathrm{e}}}{\abs{r_j}}\right)_+ r_j.
\end{align}
Substituting the definition of $\eta^{\mathrm{st}}$ and simplifying yields
\begin{equation}
    \eta_j(\rv, \lambda_{\mathrm{g}}, \lambda_{\mathrm{e}}) = \left(1-\tfrac{\lambda_{\mathrm{g}}}{\sqrt{\sum_{j'\in \cG_j} \left(\abs{r_{j'}}-\lambda_{\mathrm{e}}\right)_+^2}} \right)_+ \left( 1 - \tfrac{\lambda_{\mathrm{e}}}{\abs{r_j}} \right)_+ r_j.
\end{equation}
That is, elements are first soft-thresholded individually by $\lambda_{\mathrm{e}}$, then the entire group is zeroed out if the norm of the surviving entries falls below $\lambda_{\mathrm{g}}$.

In the real-valued setting, AMP is known to converge faster than ISTA and to achieve sharper recovery thresholds for LASSO variants~\cite{donoho2009message, bayati2011dynamics, chen2021asymptotic}.
To extend this advantage to the complex case, we apply the Onsager correction from Section~\ref{sec:main-result}, which requires the Wirtinger derivative of the denoiser $\eta$.
At points where $|r_j| > \lambda_{\mathrm{e}}$ and the group containing $j$ is active, the $(\cdot)_+$ operators can be dropped.
Let $\cG_j^* = \{j' \in \cG_j : |r_{j'}| > \lambda_{\mathrm{e}}\}$ denote the active elements in the group of $j$.
Then,
\begin{equation*}
    \eta_j(\rv, \lambda_{\mathrm{g}}, \lambda_{\mathrm{e}}) = \left(1-\tfrac{\lambda_{\mathrm{g}}}{\sqrt{\sum_{j'\in \cG_j^*} \left(\abs{r_{j'}}-\lambda_{\mathrm{e}}\right)^2}} \right) \left( 1 - \tfrac{\lambda_{\mathrm{e}}}{\abs{r_j}} \right) r_j.
\end{equation*}
Since $\eta_j$ commutes with phase rotations, i.e., $\eta_j(e^{\pzi\theta}\rv) = e^{\pzi\theta}\eta_j(\rv)$, the Onsager correction blocks inherit the $\begin{bsmallmatrix} a & b \\ b & -a \end{bsmallmatrix}$ structure required by Definition~\ref{def:canonical-transformation}.
The Wirtinger derivative simplifies to
\begin{align*}
    (\eta)'_j = \frac{\partial \eta_j}{\partial r_j}
    &= \frac{1}{2} \left(\frac{\partial \Re(\eta_j)}{\partial \Re(r_j)} + \frac{\partial \Im(\eta_j)}{\partial \Im(r_j)}\right),
\end{align*}
and $(\eta)'_j = 0$ when $j$ is inactive or belongs to an inactive group.
(The derivative is undefined on the measure-zero set where $|r_j| = \lambda_{\mathrm{e}}$ or $\norm{\eta^{\mathrm{st}}(\rv_g, \lambda_{\mathrm{e}})}_2 = \lambda_{\mathrm{g}}$; these boundary points, however, do not affect the empirical average $\langle \eta' \rangle$.)
Write $r_{g,k}$ for the $k$th element of group $g$, and let $\cS_g = \{k : |r_{g,k}| > \lambda_{\mathrm{e}}\}$ be the active indices within group $g$ (equivalently, $\cS_g$ corresponds to $\cG_j^*$ for any $j$ belonging to group $g$).
Let $\cS_a = \{g : \|\eta^{\mathrm{st}}(\rv_g, \lambda_{\mathrm{e}})\|_2 > \lambda_{\mathrm{g}}\}$ denote the set of active groups.
A direct computation gives
\begin{align*}
    &\langle \eta'(\rv, \lambda_{\mathrm{g}}, \lambda_{\mathrm{e}}) \rangle
    = \frac{1}{N} \sum_{g\in \cS_a} \Bigg( \abs{\cS_g} -\sum_{k\in \cS_g} \frac{\lambda_{\mathrm{e}}}{2\abs{r_{g,k}}} \\
    & \hspace{100pt} - \frac{\lambda_{\mathrm{g}}}{2} \frac{2\abs{\cS_g} - 1 - \sum_{k\in \cS_g} \frac{\lambda_{\mathrm{e}}}{\abs{r_{g,k}}}}{\sqrt{\sum_{k\in \cS_g} \left(\abs{r_{g,k}}-\lambda_{\mathrm{e}}\right)^2}} \Bigg) \\
    &= \frac{1}{N} \sum_{g\in \cS_a} \Bigg( \sum_{k\in \cS_g} \left( 1-\frac{\lambda_{\mathrm{e}}}{2\abs{r_{g,k}}} \right) \\
    & \hspace{40pt} - \lambda_{\mathrm{g}} \left( \frac{ \sum_{k\in \cS_g} \left(1-\frac{\lambda_{\mathrm{e}}}{2\abs{r_{g,k}}} \right)-\frac{1}{2}}{\sqrt{\sum_{k\in \cS_g} \left(\abs{r_{g,k}}-\lambda_{\mathrm{e}}\right)^2}} \right) \Bigg) \\
    &= \frac{1}{N} \sum_{g\in \cS_a} \Bigg( \frac{\lambda_{\mathrm{g}}}{2\norm{\eta^{\mathrm{st}}(\rv_g, \lambda_{\mathrm{e}})}_2} \\
    & \hspace{40pt} + \left(1-\frac{\lambda_{\mathrm{g}}}{\norm{\eta^{\mathrm{st}}(\rv_g, \lambda_{\mathrm{e}})}_2}\right) \left(\sum_{j\in \cS_g} \left(1-\frac{\lambda_{\mathrm{e}}}{2\abs{r_{g,j}}}\right) \right)  \Bigg) .
\end{align*}
Moreover, the SGL denoiser satisfies the regularity conditions of Assumption~\ref{assumption:reg-complex}.
Indeed, since complex soft-thresholding is $1$-Lipschitz and the group-norm shrinkage operator is non-expansive, their composition is Lipschitz continuous.
The required moment limits (A4)--(A6) follow from the finite second-moment assumption on~$\xv_0$.

The expressions above, together with the state evolution of Corollary~\ref{cor:complex-amp-vec-est}, provide all the ingredients needed to predict the asymptotic performance of complex AMP with the SGL denoiser.
We validate these predictions numerically in the next section.

\section{Numerical Experiments}
\label{sec:experiments}

We present numerical experiments on the complex sparse group LASSO problem formulated in Section~\ref{sec:complex_sparse_group_lasso}.
The goals are twofold: (i)~empirically validate the state evolution predictions of Corollary~\ref{cor:complex-amp-vec-est}, and (ii)~demonstrate the performance gains afforded by the proposed non-separable complex denoiser that jointly exploits group sparsity, intra-group sparsity, and the complex structure of the signal.

\subsection{Experimental Setting}

We adopt the signal model of Section~\ref{sec:complex_sparse_group_lasso} with $N = 1000$ signal elements partitioned into $G = 10$ groups, each of size $S = 100$.
The measurement matrix $\Av \in \Cb^{n \times N}$ has i.i.d. entries $\Av_{i,j} \sim \cC\cN(0,1/n)$ with $n = \delta N$, and the noise vector is $\wv \sim \cC\cN(\zerov, \sigma_w^2 \Iv_n)$.
We set the undersampling ratio $\delta = 0.6$ (so $n = 600$), the sparsity ratio $\varrho = k/n = 0.5$, and the fraction of active groups $r_{\mathrm{ga}} = 0.6$.
Therefore, the total number of active elements $k = \rho n = 0.5 \times 600 = 300$ divided among the $r_{\rm ga} G = 6$ active groups, with each active group having $k/r_{\rm ga} G = 300/6 =50$ active elements
(which gives $r_{\rm ea} = 50/(N/G) = 50/100 = 0.5$).
Each of these active elements is drawn uniformly on the complex unit circle.
All reported results are averaged over $200$ independent realizations of $(\Av, \xv_0, \wv)$.
While the state evolution results of Corollary~\ref{cor:complex-amp-vec-est} are asymptotic ($N \to \infty$), the close agreement observed at $N = 1000$ is consistent with known finite-size behavior of AMP algorithms~\cite{bayati2011dynamics}.

We measure reconstruction quality by the \emph{normalized mean-squared error},
\begin{align}
    \label{eq:nmse}
    \mathrm{NMSE} \defeq \frac{\norm{\xv^t - \xv_0}_2^2}{\norm{\xv_0}_2^2} .
\end{align}
By Corollary~\ref{cor:complex-amp-vec-est}, the asymptotic MSE satisfies $(1/N)\norm{\xv^t - \xv_0}_2^2 \to \delta(\tau_{t+1}^2 - \sigma_w^2)$, where $\tau_t$ is determined by the state evolution recursion.
In a finite-dimensional implementation, $\tau_t$ is not directly available; however, because the residual $\zv^t$ behaves as $\cC\cN(\zerov, \tau_t^2 \Iv_n)$ under the SE model, the empirical estimate $\tau_t \approx \norm{\zv^t}_2 / \sqrt{n}$ is asymptotically consistent~\cite{donoho2009message}.
We use this estimate to set the denoiser thresholds at each iteration as
\begin{xalignat}{2}
    \label{eq:threshold-scaling}
    \lambda_{\mathrm{g}} &= \alpha_g \, \tau_t, &
    \lambda_{\mathrm{e}} & = \alpha_e \, \tau_t,
\end{xalignat}
with $\alpha_g = 0.2$ and $\alpha_e = 0.8$ throughout.
In certain configurations, the sparse group LASSO (SGL) denoiser is applied only every $\kappa$ iterations, with a simpler component-wise (LASSO) denoiser used at intermediate steps.
This periodic application promotes convergence by allowing the effective observation $\rv^t$ to be sufficiently refined before the non-separable group-level shrinkage is applied.
The value of $\kappa$ is specified per experiment below.

\subsection{Convergence and State Evolution Validation}

We first examine the convergence behavior of complex AMP versus complex ISTA, both employing the non-separable SGL denoiser of Section~\ref{sec:complex_sparse_group_lasso}.
The SNR is fixed at $30$\,dB, the denoiser is applied at every iteration ($\kappa = 1$), and both algorithms are run for $200$ iterations.
\begin{figure}[t]
    \centering
    \begin{tikzpicture}
\begin{axis}[
    width=7cm,
    height=5.5cm,
    font=\small,
    xmin=0, xmax=100,
    xlabel={Number of Iterations},
    ylabel={NMSE},
    title={Convergence of Complex AMP and ISTA},
    grid=both,
    minor grid style={opacity=0.25},
    major grid style={opacity=0.5},
    legend style={
        at={(0.97,0.97)},
        anchor=north east,
        rounded corners=4pt,
        font=\footnotesize,
        cells={anchor=west}
    }
]

\addplot[color=teal, line width=1.5pt]
    table[x=iteration, y=nmse_ista] {Figures/nmse_vs_iter.dat};
\addlegendentry{Complex ISTA}

\addplot[color=purple, line width=1.5pt]
    table[x=iteration, y=nmse_amp] {Figures/nmse_vs_iter.dat};
\addlegendentry{Complex AMP}

\addplot[color=gray, line width=1.5pt, dashed]
    table[x=iteration, y=se_amp] {Figures/nmse_vs_iter.dat};
\addlegendentry{State Evolution}

\end{axis}
\end{tikzpicture}
    \caption{Convergence of complex AMP and complex ISTA with the non-separable SGL denoiser.
    The dashed curve shows the state evolution (SE) prediction from Corollary~\ref{cor:complex-amp-vec-est}.
    Parameters: $\delta = 0.6$, $\varrho = 0.5$, $r_{\mathrm{ga}} = 0.6$, $\mathrm{SNR} = 30$\,dB, $\alpha_g = 0.2$, $\alpha_e = 0.8$.}
    \label{fig:nmse_vs_iter}
\end{figure}
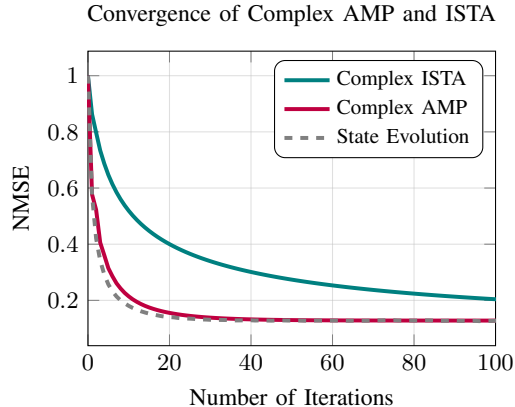

Fig.~\ref{fig:nmse_vs_iter} shows the NMSE as a function of iteration count.
Complex AMP converges rapidly, within approximately $15$ iterations, and its trajectory is accurately predicted by the state evolution, validating the theoretical framework of Corollary~\ref{cor:complex-amp-vec-est}.
In this regime, complex ISTA appears to approach the same fixed point but at a significantly slower rate.
This is consistent with the well-known advantage of AMP over ISTA in the real-valued setting~\cite{donoho2009message, bayati2011dynamics}.
Unlike AMP, ISTA lacks the Onsager correction and therefore does not admit a state evolution characterization.

\subsection{Performance Comparison}

We next compare the recovery performance of four AMP variants as a function of SNR, ranging from $10$ to $40$\,dB.
The four methods differ in how they exploit the signal structure through their choice of denoiser:
\begin{itemize}
    \item \emph{Complex LASSO}: element-wise soft thresholding with $\lambda_{\mathrm{e}} = \alpha_e \, \tau_t$ and $\lambda_{\mathrm{g}} = 0$, which treats all entries as independently sparse and ignores the group structure;
    \item \emph{Complex GL}: group thresholding with $\lambda_{\mathrm{g}} = \alpha_g \, \tau_t$ and $\lambda_{\mathrm{e}} = 0$, which exploits only group-level sparsity;
    \item \emph{Real SGL}: the SGL denoiser applied to the real and imaginary parts independently, thereby capturing both sparsity levels but treating the complex components as separate real-valued problems;
    \item \emph{Complex SGL}: the proposed non-separable complex SGL denoiser from Section~\ref{sec:complex_sparse_group_lasso}, which jointly accounts for group sparsity, intra-group sparsity, and the complex nature of the signal.
\end{itemize}
For the Real and Complex SGL variants, the SGL denoiser is applied every $\kappa = 5$ iterations to promote convergence stability, with the simpler component-wise denoiser used at intermediate steps.

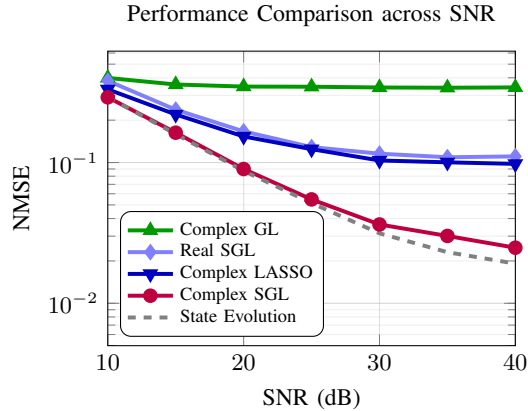
\begin{figure}[t]
    \centering
    \begin{tikzpicture}
\begin{semilogyaxis}[
    width=7cm,
    height=5.5cm,
    font=\small,
    xlabel={SNR (dB)},
    xmin=10, xmax=40,
    ylabel={NMSE},
    ymin=0.005,
    title={Performance Comparison across SNR},
    grid=both,
    minor grid style={opacity=0.25},
    major grid style={opacity=0.5},
    legend style={
        at={(0.03,0.03)},
        anchor=south west,
        rounded corners=4pt,
        font=\scriptsize,
        cells={anchor=west},
        inner sep=2pt,
        row sep=-2pt
    }
]

\addplot[color=green!60!black, line width=1.5pt, mark=triangle*]
    table[x=snr_db, y=nmse_group_lasso] {Figures/nmse_vs_snr.dat};
\addlegendentry{Complex GL}

\addplot[color=blue!50, line width=1.5pt, mark=diamond*]
    table[x=snr_db, y=nmse_real_sgl] {Figures/nmse_vs_snr.dat};
\addlegendentry{Real SGL}

\addplot[color=blue!75!black, line width=1.5pt, mark=triangle*, mark options={rotate=180}]
    table[x=snr_db, y=nmse_lasso] {Figures/nmse_vs_snr.dat};
\addlegendentry{Complex LASSO}

\addplot[color=purple, line width=1.5pt, mark=*]
    table[x=snr_db, y=nmse_sgl, y error=std_sgl] {Figures/nmse_vs_snr.dat};
\addlegendentry{Complex SGL}

\addplot[color=gray, line width=1.5pt, dashed]
    table[x=snr_db, y=se_sgl] {Figures/nmse_vs_snr.dat};
\addlegendentry{State Evolution}

\end{semilogyaxis}
\end{tikzpicture}
    \caption{NMSE versus SNR for four AMP variants.
    The dashed curve shows the SE prediction for Complex SGL.
    Parameters: $\delta = 0.6$, $\varrho = 0.5$, $r_{\mathrm{ga}} = 0.6$, $\alpha_g = 0.2$, $\alpha_e = 0.8$, $\kappa = 5$.}
    \label{fig:nmse_vs_snr}
\end{figure}

Fig.~\ref{fig:nmse_vs_snr} presents the NMSE as a function of SNR for the four methods.
Complex GL performs worst across all SNR values; by exploiting only the group-level sparsity, it leaves the intra-group sparsity structure entirely unexploited.
Complex LASSO improves upon GL by penalizing individual entries, but it ignores the group structure and therefore yields suboptimal performance.
Real SGL jointly promotes both levels of sparsity; however, by operating on the real and imaginary components independently, it fails to exploit the coupling between them that is inherent to complex-valued signals.
In contrast, Complex SGL fully accounts for all three structural properties and achieves the best NMSE across the entire SNR range.
The state evolution prediction closely tracks the empirical NMSE of Complex SGL, providing further validation of the theoretical framework developed in Section~\ref{sec:main-result}.

\section{Conclusion}

This article proved state evolution for complex-valued AMP with non-separable denoising functions.
The approach constructs an augmented real matrix-valued system that lifts the representation from $2n$ to $4n$ real dimensions, then recovers the complex domain through a many-to-one canonical transformation.
This strategy leverages existing real-valued matrix AMP results~\cite{berthier2020state,gerbelot2023amp} and yields a direct reduction from the $2 \times 2$ matrix-valued state evolution of the augmented system to scalar complex state evolution recursions.
The framework further extends to the complex matrix-valued setting, where multiple feature vectors are estimated simultaneously and the Onsager corrections generalize from complex scalars to $p \times p$ complex Jacobian matrices.

A key insight is that Wirtinger derivatives arise in the Onsager correction as a direct byproduct of the canonical transformation, unifying the treatment of complex derivatives within the AMP framework rather than requiring them to be imposed a priori.
The transformation and its structural properties were shown to preserve the AMP dynamics, enabling a complete state evolution characterization (Theorem~\ref{thm:complex-amp-main}).
To demonstrate the practical reach of the theory, it was applied to the complex SGL, for which a closed-form proximal operator and its Wirtinger derivative were derived; this denoiser jointly exploits group-level and element-level sparsity in the complex domain.
Numerical experiments confirmed that the state evolution predictions closely track the empirical NMSE across the full SNR range, and that the proposed complex SGL denoiser uniformly outperforms baselines that exploit only partial structure (complex LASSO, complex GL, and real-valued SGL).

The framework applies broadly to complex-valued inverse problems in which a non-separable denoiser can exploit structural dependencies that separable or real-decoupled approaches cannot.
Beyond the SGL, natural application domains include OTFS preamble detection for unsourced random access~\cite{mirri2025approximate}, OFDM channel estimation with clustered sparsity, and array processing problems where phase coherence and spatial group structure are central.

The current theory assumes i.i.d. complex Gaussian sensing matrices, which limits direct applicability to problems with structured operators.
Several directions for future work address this and other extensions.
First, extending the framework to structured sensing operators such as partial DFT matrices, OTFS channel models, or random convolutional designs, through connections to orthogonal or vector AMP, would broaden its applicability.
Second, because state evolution provides a scalar recursion predicting the mean-squared error, it can serve as an oracle for optimizing regularization parameters and the denoiser application schedule, automating choices that are currently set heuristically.
Third, unrolling the complex AMP iterations into a learnable architecture (deep unfolding) is a natural extension; the state evolution characterization proved here provides the theoretical foundation for such designs in the complex non-separable setting.
Fourth, replacing the proximal denoiser with a Bayes-optimal posterior mean estimator, which is itself non-separable under group-structured priors, may strengthen recovery guarantees.
Finally, extending the framework to generalized linear observation models (e.g., quantized measurements or phase retrieval) is a further direction for broadening the scope of complex AMP.

\section*{Acknowledgment}
This material is based, in part, upon work conducted at Texas A\&M University and supported by the National Science Foundation (NSF) under Grant CNS-2148354.
This work was supported, in part, by the European Union--Next Generation EU under the Italian National Recovery and Resilience Plan (NRRP), partnership on ``Telecommunications of the Future'' (PE00000001--program ``RESTART'').

\appendices

\section{Detailed Proofs}

\label{app:detailed_proofs}

The appendix is organized as follows.
We first state the regularity assumptions of the real matrix AMP framework~\cite{gerbelot2023amp} (Assumption~\ref{assumption:reg-real}) and verify that they are implied by the complex regularity conditions of Assumption~\ref{assumption:reg-complex} under the canonical transformation.
We then define the real state evolution (Definition~\ref{def:state-evol-real}) and show in Lemma~\ref{lem:CollapsedState} that the complex state evolution of Definition~\ref{def:state-evol-complex} arises by collapsing the real $2\times 2$ arrays to complex scalars.
Finally, we prove Theorem~\ref{thm:complex-amp-main} by combining the real matrix AMP theorem with pseudo-Lipschitz closure, and derive Corollary~\ref{cor:complex-amp-vec-est} as a specialization to vector estimation.

\subsection{Regularity Transfer}
\label{sec:regularity-transfer}
\begin{assumption}[Regularity for Real AMP]
\label{assumption:reg-real}
Consider a sequence of recursions, each of the form \eqref{eq:real-mat-1} and \eqref{eq:real-mat-2}, indexed by $m$ and $M \equiv M(m)$ for which $m/M \rightarrow \Delta \in (0, \infty)$ as $m \rightarrow \infty$.
In this regime, assume the following holds.
\begin{list}{(A\arabic{enumi})}
  {\usecounter{enumi}
  \setlength{\leftmargin}{3em}
  \setlength{\labelwidth}{3.5em}
  \setlength{\labelsep}{0.5em}}
\item $\Wv_{i,j} \equiv \Wv_{i,j}(m) \stackrel{\mathrm{i.i.d.}}{\sim} \mathcal{N} \left( 0, 1/m \right)$.
\item The sequences of functions $\{f^k\}$ and $\{g^k\}$ are uniformly Lipschitz.
\item For $\Hv^0 \equiv \Hv^0(m)$, $\left\| \Hv^0 \right\|_{\mathrm{F}} / \sqrt{m}$ converges to a finite constant as $m \rightarrow \infty$.
\item The following limit exists and is finite:
\begin{equation*}
\lim_{m \rightarrow \infty} \frac{1}{m} f^0 \left( \Hv^0 \right)^{\transpose} f^0 \left( \Hv^0 \right) \in \mathbb{R}^{p \times p}
\end{equation*}
\item For any positive definite matrix $\Sv \in \mathcal{S}_p^+$, the following limit exists and is finite:
\begin{equation} \label{eq:A5}
\lim_{m \rightarrow \infty} \frac{1}{m} \mathbb{E} \left[ f^0 \left( \Hv^0 \right)^{\transpose} f^0 \left( \Zv \right) \right] \in \mathbb{R}^{p \times p}
\end{equation}
where $\Zv \in \mathbb{R}^{M \times p}$, $\Zv \sim \mathcal{N} (\zerov, \Sv \otimes \IDm_M )$.
\item For any $r,s > 0$ and any positive definite matrices $\Sv_f \in \mathcal{S}_{2p}^+$ and $\Sv_g \in \mathcal{S}_{2q}^+$, the following limits exist and are finite:
\begin{gather}
\lim_{m \rightarrow \infty} \frac{1}{m} \mathbb{E} \left[ f^r \left( \Zv_{\tau}^r \right)^{\transpose} f^s \left( \Zv_{\tau}^s \right) \right] \in \mathbb{R}^{p \times p} \label{eq:A6a} \\
\lim_{m \rightarrow \infty} \frac{1}{m} \mathbb{E} \left[ g^r \left( \Zv_{\sigma}^r \right)^{\transpose} g^s \left( \Zv_{\sigma}^s \right) \right] \in \mathbb{R}^{q \times q} \label{eq:A6b}
\end{gather}
where $\big( \Zv_{\tau}^r, \Zv_{\tau}^s \big) \in \left( \mathbb{R}^{M \times p} \right)^2$, $\big( \Zv_{\tau}^r, \Zv_{\tau}^s \big) \sim \mathcal{N} (0, \Sv_f \otimes \IDm_M )$ and, likewise, $\big( \Zv_{\sigma}^r, \Zv_{\sigma}^s \big) \in \left( \mathbb{R}^{m \times q} \right)^2$, $\big( \Zv_{\sigma}^r, \Zv_{\sigma}^s \big) \sim \mathcal{N} (\zerov, \Sv_g \otimes \IDm_m )$.
\end{list}
\end{assumption}

\begin{lemma}[Complex Regularity Implies Real Regularity in Lifted Space]
    \label{lem:complex-reg-real-lifted}
Let $f$ and $g$ be defined in terms of $\widetilde{f}$ and $\widetilde{g}$ via the canonical
transformation of Definition~\ref{def:canonical-transformation}. Then
Assumption~\ref{assumption:reg-complex} implies Assumption~\ref{assumption:reg-real}.
\end{lemma}
\begin{IEEEproof}
Under the canonical transformation of Definition~\ref{def:canonical-transformation}, the real system has parameters $m = 2n$, $M = N$, $p = q = 2$, and $\Delta = m/M = 2\delta$.
Keeping this in mind, we rapidly verify each condition of Assumption~\ref{assumption:reg-real} below.

\emph{(A1)}
Each entry $\Av_{i,j} \sim \cC\cN(0, 1/n)$ decomposes as $\Re \Av_{i,j}, \Im \Av_{i,j} \stackrel{\mathrm{i.i.d.}}{\sim} \cN(0, 1/(2n))$.
Since $\Wv$ stacks $\Re \Av_{\iota,:}$ and $\Im \Av_{\iota,:}$ as separate rows, its entries are i.i.d.\ $\cN(0, 1/(2n)) = \cN(0, 1/m)$.

\emph{(A2)}
The $\iota$th row of $f^k$ is defined as $f^k_\iota(\Hv^k) = \Mv^k_\iota = [\Re \pzf^k_\iota(\Hv^k \rho), \Im \pzf^k_\iota(\Hv^k \rho)]$.
This is a composition of the linear map $\Hv^k \mapsto \Hv^k \rho$ (operator norm $\|\rho\| = \sqrt{2}$), the complex function $\pzf^k_\iota$ (uniformly Lipschitz by complex~(A2)), and the real--imaginary extraction $z \mapsto (\Re z, \Im z)$ (1-Lipschitz).
Since each factor has controlled Lipschitz growth, so does each $f^k_\iota$ (with a Lipschitz constant uniform in $\iota$), and therefore $f^k$.
The $\iota$th block of $g^k$ is defined as $g_\iota^k(\Ev^k) = \Qv^k_\iota =
\begin{bsmallmatrix}
    \Re \pzg_\iota^k(\rho^\transpose \Ev^k \rho) & \Im \pzg_\iota^k(\rho^\transpose \Ev^k \rho) \\
    \Im \pzg_\iota^k(\rho^\transpose \Ev^k \rho) & -\Re \pzg_\iota^k(\rho^\transpose \Ev^k \rho)
\end{bsmallmatrix}$.
This is a composition of the linear map $\Ev^k \mapsto \Rv \Ev^k \rho$ (operator norm $2$), the complex function $\pzg^k_\iota$ (uniformly Lipschitz by complex~(A2)), and the extraction $z \mapsto \begin{bsmallmatrix} \Re z & \Im z \\ \Im z & -\Re z \end{bsmallmatrix}$ (2-Lipschitz).
Since each factor has controlled Lipschitz growth, so does each $g^k_\iota$ (with a Lipschitz constant uniform in $\iota$), and therefore $g^k$.

\emph{(A3)}
By Definition~\ref{def:canonical-transformation}, $\Hv^0_\iota = [\Re \uv^0_\iota, \Im \uv^0_\iota]$, so $\|\Hv^0\|_{\mathrm{F}}^2 = \sum_\iota (|\Re \uv^0_\iota|^2 + |\Im \uv^0_\iota|^2) = \|\uv^0\|_2^2$.
Hence $\|\Hv^0\|_{\mathrm{F}} / \sqrt{m} = \|\uv^0\|_2 / \sqrt{2n}$, which converges by complex~(A3).

\emph{(A4)}
By Definition~\ref{def:canonical-transformation}, $\Mv^0_\iota = [\Re \pzf^0_\iota(\uv^0), \Im \pzf^0_\iota(\uv^0)]$, so
\begin{equation*}
\frac{1}{m} f^0(\Hv^0)^\transpose f^0(\Hv^0)
= \frac{1}{2n} \sum_\iota (\Mv^0_\iota)^\transpose \Mv^0_\iota \in \Rb^{2\times 2} .
\end{equation*}
The three independent entries of this symmetric matrix are linear combinations of $\langle \pzf^0(\uv^0), \pzf^0(\uv^0) \rangle / n$ (which is real), $\Re \langle \pzf^0(\uv^0), \overline{\pzf^0(\uv^0)} \rangle / n$, and $\Im \langle \pzf^0(\uv^0), \overline{\pzf^0(\uv^0)} \rangle / n$, all of which exist by complex~(A4).

\emph{(A5)--(A6)}
The real limits in \eqref{eq:A5}--\eqref{eq:A6b} involve Gram matrices of the form
\begin{xalignat*}{2}
\frac{1}{m} & \Eb[f^r(\zetav_\tau^r)^\transpose f^s(\zetav_\tau^s)] &
\frac{1}{m} & \Eb[g^r(\zetav_\sigma^r)^\transpose g^s(\zetav_\sigma^s)] ,
\end{xalignat*}
where the arguments are jointly real Gaussian.
By the same block-wise expansion as in~(A4), the entries of each $2 \times 2$ Gram matrix are determined by $\langle \pzf^r(\Zv_\tau^r), \pzf^s(\Zv_\tau^s) \rangle / n$, $\langle \pzf^r(\Zv_\tau^r), \overline{\pzf^s(\Zv_\tau^s)} \rangle / n$ (and likewise for $\pzg$), where $\Zv_\tau^r = \zetav_\tau^r \rho$ is the complex collapse.
The existence of these complex limits is guaranteed by complex~(A5)--(A6).
\end{IEEEproof}

\subsection{State Evolution Collapse}
\label{sec:state-evolution-collapse}
With Assumption~\ref{assumption:reg-real} in hand, we define the state evolution for the real-matrix AMP.

\begin{definition}[Real State Evolution] \label{def:state-evol-real}
Define two doubly infinite arrays, $\left( \Sigmav_{s, r} \right)_{s, r \geq 0}$ and $\left( \Tv_{s, r} \right)_{s, r \geq 1}$ through the following recursion, known as \emph{state evolution}.
Let the first state evolution iterate be
\begin{equation}
\label{eq:init-real}
\Sigmav^{0,0} =  \lim_{m \rightarrow \infty} \frac{1}{m} f^0 \left( \Hv^0 \right)^{\transpose} f^0 \left( \Hv^0 \right) \in \Rb^{p\times p}.
\end{equation}
Once $\left( \Sigmav^{s, r} \right)_{0 \leq s, r \leq t}$ and $\left( \Tv^{s, r} \right)_{1 \leq s, r \leq t}$ are specified, recursively define state evolution iterates for $t+1$ as follows.

Set $\left( \zetav_{\sigma}^0, \ldots, \zetav_{\sigma}^t \right) \sim \mathcal{N} \left( \zerov, \left( \Sigmav^{s, r} \right)_{0 \le s, r \leq t} \otimes \IDm_m \right)$, and define for $s\in\{0, \dots, t\}$:
\begin{equation}
\label{eq:tau-real}
(\Tv^{t+1, s+1})^\transpose = \Tv^{s+1, t+1}
= \lim_{m \rightarrow \infty} \frac{1}{m} \mathbb{E} \left[ g^s \left( \zetav_{\sigma}^s \right)^{\transpose} g^t \left( \zetav_{\sigma}^t \right) \right].
\end{equation}
This characterizes array $\left( \Tv^{s, r} \right)_{1 \leq s, r \leq t+1}$.
Then take $\zetav_{\tau}^0 = \Hv^0$, set $\left( \zetav_{\tau}^1, \ldots, \zetav_{\tau}^{t+1} \right) \sim \mathcal{N} \left( \zerov, \left( \Tv_{s, r} \right)_{1\le s, r \leq t+1} \otimes \IDm_M \right)$, and define for $s\in\{0, \dots, t+1\}$
\begin{equation}
\label{eq:sigma-real}
(\Sigmav^{t+1, s})^\transpose = \Sigmav^{s, t+1}
= \lim_{m \rightarrow \infty} \frac{1}{m} \mathbb{E} \left[ f^s \left( \zetav_{\tau}^s \right)^{\transpose} f^{t+1} \left( \zetav_{\tau}^{t+1} \right) \right].
\end{equation}
This characterizes array $\left( \Sigmav^{s, r} \right)_{0 \leq s, r \leq t+1}$.
\end{definition}

\begin{lemma}[Collapsed Real State Evolution Implies Complex State Evolution]
\label{lem:CollapsedState}
Let $f$ and $g$ be defined in terms of $\pzf$ and $\pzg$ via the canonical transformation of Definition~\ref{def:canonical-transformation}. Then the real state evolution equations~\eqref{eq:init-real}, \eqref{eq:tau-real}, and \eqref{eq:sigma-real} imply the complex state evolution equations~\eqref{eq:init-complex}, \eqref{eq:tau-complex}, and \eqref{eq:sigma-complex}.
\end{lemma}
\begin{IEEEproof}
Under the canonical transformation, $\Wv \in \Rb^{2n \times N}$, so the real framework of Definition~\ref{def:state-evol-real} operates with $m = 2n$, $M = N$, and $p = q = 2$.
The strategy is to show that each real $2 \times 2$ state evolution array ($\Sigmav^{s,r}$ and $\Tv^{s,r}$) collapses to the corresponding complex scalar ($\sigma^{s,r}$ and $\tau^{s,r}$) under the canonical transformation, and that the associated Gaussian distributions are proper complex normal (i.e., the relation matrix vanishes, $C_{s,t} = 0$).
We first recall the relationship between real and complex Gaussian covariances, then apply it to the $\Tv$ and $\Sigmav$ arrays in turn.

Suppose that $\xiv_{s, t} \triangleq [\Re Z_s, \Im Z_s, \Re Z_t, \Im Z_t]^\transpose \in \Rb^4$ is a zero-mean Gaussian random vector.
The corresponding complex representation $\Zv_{s, t} = [Z_s, Z_t]^\transpose \in \Cb^2$ has covariance and relation entries
\begin{align*}
    \Gamma_{s, t} = \Eb[Z_s \overline{Z_t}]
    &= (\Eb[\Re Z_s \Re Z_t] + \Eb[\Im Z_s \Im Z_t])  \\
    & \hspace{20pt} + \pzi (\Eb[\Im Z_s \Re Z_t ] - \Eb[\Re Z_s \Im Z_t]), \\
    C_{s, t} = \Eb[Z_s Z_t]
    &= (\Eb[\Re Z_s \Re Z_t] - \Eb[\Im Z_s \Im Z_t]) \\
    & \hspace{20pt} + \pzi (\Eb[\Im Z_s \Re Z_t ] + \Eb[\Re Z_s \Im Z_t]).
\end{align*}
Thus, both $\Gamma_{s,t}$ and $C_{s,t}$ are fully determined by the four real second moments $\Eb[\Re Z_s \Re Z_t]$, $\Eb[\Im Z_s \Im Z_t]$, $\Eb[\Im Z_s \Re Z_t]$, and $\Eb[\Re Z_s \Im Z_t]$.

We begin with the initialization~\eqref{eq:init-real}.
By the same computation as in~(A4) of Lemma~\ref{lem:complex-reg-real-lifted}, we have that $\Sigmav^{0,0} = (1/(2n)) \sum_\iota (\Mv^0_\iota)^\transpose \Mv^0_\iota \in \Rb^{2\times 2}$, whose entries collapse to $\sigma^{0,0} = (1/\delta) \lim_{N\to\infty} (1/N) \langle \pzf^0(\uv^0), \pzf^0(\uv^0) \rangle$, matching~\eqref{eq:init-complex}.

We now turn to the recursive arrays $\Tv$ and $\Sigmav$.
From~\eqref{eq:tau-real} with $m = 2n$, we have
\begin{align*}
    \Tv^{s+1, t+1} &= \lim_{n \to \infty} \frac{1}{2n} \Eb[g^s(\zetav_\sigma^s)^\transpose g^t(\zetav_\sigma^t)].
\end{align*}
By Definition~\ref{def:canonical-transformation}, $g^s(\zetav_\sigma^s)$ decomposes into $n$ blocks $\Qv^s_\iota = \begin{bsmallmatrix} \Re \pzg_\iota^s & \Im \pzg_\iota^s \\ \Im \pzg_\iota^s & -\Re \pzg_\iota^s \end{bsmallmatrix}$, where $\pzg^s_\iota \equiv \pzg^s_\iota(\Rv \zetav_\sigma^s \rho)$, and similarly for $g^t(\zetav_\sigma^t)$.
Expanding the product block-wise, we obtain
\begin{align*}
    &\Tv^{s+1, t+1}
    = \lim_{n \to \infty} \frac{1}{2n} \Eb \sum_{\iota=1}^n (\Qv_\iota^s)^\transpose \Qv_\iota^t \\
    &= \lim_{n \to \infty} \frac{1}{2n} \Eb \sum_{\iota=1}^n
    \begin{bmatrix}
         \Re \pzg_\iota^s \Re \pzg_\iota^t +  \Im \pzg_\iota^s \Im \pzg_\iota^t & \Re \pzg_\iota^s \Im \pzg_\iota^t -  \Re \pzg_\iota^t \Im \pzg_\iota^s  \\
         \Re \pzg_\iota^t \Im \pzg_\iota^s -  \Re \pzg_\iota^s \Im \pzg_\iota^t &  \Re \pzg_\iota^s \Re \pzg_\iota^t +  \Im \pzg_\iota^s \Im \pzg_\iota^t
    \end{bmatrix}.
\end{align*}
Since
\begin{equation}
    \begin{split}
    \langle \pzg^s_\iota, \pzg_\iota^t \rangle &\triangleq  \pzg_\iota^s \overline{\pzg^t_\iota}
    = (\Re \pzg^s_\iota + \pzi \Im \pzg^s_\iota) (\Re \pzg^t_\iota - \pzi \Im \pzg^t_\iota) \\
    &= (\Re \pzg_\iota^s \Re \pzg_\iota^t +  \Im \pzg_\iota^s \Im \pzg_\iota^t) + \pzi (\Re \pzg_\iota^t \Im \pzg_\iota^s -  \Re \pzg_\iota^s \Im \pzg_\iota^t), \label{eq:expand-complex-inner-prod}
    \end{split}
\end{equation}
the matrix entries are the real and imaginary parts of $\langle \pzg^s_\iota, \pzg^t_\iota \rangle$.
Recalling $\tau^{s+1, t+1}$ from~\eqref{eq:tau-complex}, we obtain
\begin{equation*}
\Tv^{s+1, t+1} =
    \begin{bmatrix}
        \frac{1}{2} \Re \tau^{s+1, t+1} & -\frac{1}{2} \Im \tau^{s+1, t+1} \\
        \frac{1}{2} \Im \tau^{s+1, t+1} & \frac{1}{2} \Re \tau^{s+1, t+1}
    \end{bmatrix}.
\end{equation*}
Since $\zetav_\tau^s \in \Rb^{N \times 2}$ has rows in $\Rb^{1\times 2}$, the collapse $\Zv_{\tau, \iota}^s = \zetav_{\tau, \iota}^s \rho$ gives $\Re \Zv_{\tau,\iota}^s = \zeta_{\tau,\iota,1}^s$ and $\Im \Zv_{\tau,\iota}^s = \zeta_{\tau,\iota,2}^s$, so the entries of $\Tv^{s+1,t+1}$ are directly the four real second moments.
Substituting into the covariance and relation formulas, we obtain
\begin{align*}
    \Gamma_{s, t} &= \Tv^{s+1, t+1}_{1,1} + \Tv^{s+1, t+1}_{2,2} \\
    & \hspace{20pt} + \pzi \left(\Tv^{s+1, t+1}_{2,1} -  \Tv^{s+1, t+1}_{1,2} \right) = \tau^{s+1, t+1}, \\
    C_{s,t} &= \Tv^{s+1, t+1}_{1,1} - \Tv^{s+1, t+1}_{2,2}   \\
    & \hspace{20pt} + \pzi \left( \Tv^{s+1, t+1}_{2,1} + \Tv^{s+1, t+1}_{1,2} \right) = 0.
\end{align*}
Collecting these results, we obtain $(\Zv_\tau^1, \dots, \Zv_{\tau}^{t+1}) \sim \cC\cN(\zerov, (\tau^{r, s})_{1\le r,s \le t+1} \otimes \Iv_N)$.

Next, we must establish the relation between $\zetav_\sigma^s, \zetav_\sigma^t$ and $Z_\sigma^s, Z_\sigma^t$.
In what follows we drop the subscript $\sigma$ for ease of notation.
Each block $\zetav^s_\iota \in \Rb^{2\times 2}$ has two rows $\zetav^s_{\iota,1,:}, \zetav^s_{\iota,2,:} \in \Rb^{1\times 2}$ that are i.i.d.\ across rows (by the $\otimes \IDm_m$ Kronecker structure), with joint covariance across iterations given by $(\Sigmav^{s,r})$.
The complex collapse $Z^s_\iota = \rho^\transpose \zetav^s_\iota \rho$ yields
\begin{align*}
    \Re Z^s = \zeta_{1,1}^s - \zeta_{2,2}^s, \quad \Im Z^s = \zeta_{1,2}^s + \zeta_{2,1}^s.
\end{align*}

Using the independence of the two rows, their identical distribution, zero mean, and the symmetry $(\Sigmav^{s,t})^\transpose = \Sigmav^{t,s}$ from~\eqref{eq:sigma-real}, we have
\begin{align*}
    \Eb[\Re Z^s \Re Z^t] &= \Eb[\zeta_{1,1}^s \zeta_{1,1}^t]+\Eb[\zeta_{2,2}^s\zeta_{2,2}^t] \\
    & \hspace{20pt} - \Eb[\zeta_{1,1}^s \zeta_{2,2}^t] -\Eb[\zeta_{1,1}^t \zeta_{2,2}^s] \\
    &= \Eb[\zeta_{1,1}^s \zeta_{1,1}^t] + \Eb[\zeta_{1,2}^s\zeta_{1,2}^t] \\
    &\hspace{20pt} - \underbrace{\Eb[\zeta_{1,1}^s] \Eb[\zeta_{2,2}^t]}_{=0} - \underbrace{\Eb[\zeta_{1,1}^t] \Eb[\zeta_{2,2}^s]}_{=0} \\
    &= \Eb[\zeta_{1,1}^s \zeta_{1,1}^t]+\Eb[\zeta_{1,2}^s\zeta_{1,2}^t] = \Sigmav^{s,t}_{1,1} + \Sigmav^{s,t}_{2,2}, \\
    \Eb[\Im Z^s \Im Z^t] &= \Sigmav^{s,t}_{1,1} + \Sigmav^{s,t}_{2,2} = \Eb[\Re Z^s \Re Z^t], \\
    \Eb[\Im Z^s \Re Z^t] &= \Eb[\zeta_{1,2}^s \zeta_{1,1}^t]-\Eb[\zeta_{2,1}^s \zeta_{2,2}^t] \\
    &= \Eb[\zeta_{1,2}^s \zeta_{1,1}^t]-\Eb[\zeta_{1,1}^s \zeta_{1,2}^t] = \Sigmav^{t,s}_{1,2} - \Sigmav^{s,t}_{1,2}, \\
    \Eb[\Re Z^s \Im Z^t] &= \Eb[\zeta_{1,1}^s \zeta_{1,2}^t]-\Eb[\zeta_{2,1}^t \zeta_{2,2}^s] \\
    &= \Eb[\zeta_{1,1}^s \zeta_{1,2}^t]-\Eb[\zeta_{1,1}^t \zeta_{1,2}^s] \\
    &= \Sigmav^{s,t}_{1,2} - \Sigmav^{t,s}_{1,2} = -\Eb[\Im Z^s \Re Z^t].
\end{align*}
Substituting into the covariance and relation formulas yields
\begin{align*}
    \Gamma_{s,t} &= 2 \left(\left(\Sigmav^{s,t}_{1,1} + \Sigmav^{s,t}_{2,2}\right) + \pzi \left(\Sigmav^{t,s}_{1,2} - \Sigmav^{s,t}_{1,2} \right)\right), \\
    C_{s,t} &= \left(\Eb[\Re Z^s \Re Z^t] - \Eb[\Im Z^s \Im Z^t]\right) \\
    & \hspace{20pt} + \pzi \left(\Eb[\Im Z^s \Re Z^t] + \Eb[\Re Z^s \Im Z^t]\right) = 0,
\end{align*}
where the first and the second parentheses in $C_{s,t}$ vanish because $\Eb[\Re Z^s \Re Z^t] = \Eb[\Im Z^s \Im Z^t]$ and $\Eb[\Re Z^s \Im Z^t] = -\Eb[\Im Z^s \Re Z^t]$, respectively.
Next, we express $\Gamma_{s,t}$ in terms of $\sigma^{s,t}$.
By~\eqref{eq:sigma-real} with $m = 2n$ and Definition~\ref{def:canonical-transformation}, $f^s(\zetav_\tau^s)$ has $N$ rows $\Mv^s_\iota = [\Re \pzf^s_\iota, \Im \pzf^s_\iota]$ where $\pzf^s_\iota \equiv \pzf^s_\iota(\zetav_\tau^s \rho)$, giving
\begin{align*}
\Sigmav^{s,t}
&= \lim_{n\to \infty} \frac{1}{2n} \Eb \sum_{\iota=1}^N (\Mv^s_\iota)^\transpose \Mv^t_\iota \\
&= \lim_{n\to \infty} \frac{1}{2n} \Eb \sum_{\iota=1}^N \begin{bmatrix}
        \Re \pzf^s_\iota \Re \pzf^t_\iota & \Re \pzf^s_\iota \Im \pzf^t_\iota \\
        \Im \pzf^s_\iota \Re \pzf^t_\iota & \Im \pzf^s_\iota \Im \pzf^t_\iota
    \end{bmatrix}.
\end{align*}
Recalling $\sigma^{s, t}$ from~\eqref{eq:sigma-complex} and applying~\eqref{eq:expand-complex-inner-prod},
\begin{align*}
    2 \left(\Sigmav^{s,t}_{1,1} + \Sigmav^{s,t}_{2,2}\right) = \Re \sigma^{s,t}, \quad 2\left(\Sigmav^{t,s}_{1,2} - \Sigmav^{s,t}_{1,2} \right) = \Im \sigma^{s,t},
\end{align*}
whence we have $\Gamma_{s,t} = \sigma^{s,t}$.
It follows that $(\Zv_{\sigma}^0, \dots, \Zv_{\sigma}^t) \sim \cC\cN(\zerov, (\sigma^{s,r})_{0\le s, r \le t} \otimes \Iv_n)$.
\end{IEEEproof}

\subsection{Proof of Theorem~\ref{thm:complex-amp-main}}
\label{sec:pseudo-lipschitz-functions}
\begin{IEEEproof}
Given a complex-valued pseudo-Lipschitz function $\Phi_n: (\Cb^N\times \Cb^n)^t \times \Cb^N \to \Cb$, define
\begin{align*}
    &\Psi_n(\Hv^0, \Ev^0, \dots, \Hv^{t-1}, \Ev^{t-1}, \Hv^t) \\
    &=
    \begin{bmatrix}
        \Re \Phi_n(\Hv^0 \rho, \Rv \Ev^0 \rho, \dots, \Hv^{t-1} \rho, \Rv \Ev^{t-1} \rho, \Hv^t \rho) \\
        \Im \Phi_n(\Hv^0 \rho, \Rv \Ev^0 \rho, \dots, \Hv^{t-1} \rho, \Rv \Ev^{t-1} \rho, \Hv^t \rho)
    \end{bmatrix}^\transpose{\in}\Rb^{1\times 2}.
\end{align*}
Since $\Psi_n$ is a composition of the linear maps $\Hv^k \mapsto \Hv^k \rho$ and $\Ev^k \mapsto \Rv \Ev^k \rho$ (both bounded, with operator norms $\|\rho\| = \sqrt{2}$ and $\|\Rv\|\|\rho\| = 2$), the pseudo-Lipschitz function $\Phi_n$, and the $1$-Lipschitz extraction $z \mapsto [\Re z, \Im z]^\transpose$, the function $\Psi_n$ is pseudo-Lipschitz of the same order.
Recalling~\cite[Thm. 2]{gerbelot2023amp}, we have that for any sequence of pseudo-Lipschitz functions $\Psi_n:(\Rb^{N \times 2}\times \Rb^{2n\times 2}) \times \Rb^{N\times 2} \to \Rb^{1\times 2}$ and $\Hv^0$, we have
\begin{align*}
    &\Psi_n(\Hv^0, \Ev^0, \dots, \Hv^{t-1}, \Ev^{t-1}, \Hv^t) \\
    & \hspace{20pt} \stackrel{\mathrm{P}}{\simeq}  \Psi_n(\Hv^0, \zetav_\sigma^0, \dots, \zetav_\tau^{t-1}, \zetav_\sigma^{t-1}, \zetav_\tau^t)
\end{align*}
Since $\Psi_n = [\Re \Phi_n, \Im \Phi_n]^\transpose$, we recover $\Phi_n = \Psi_n \rho$.
Writing $\uv^k = \Hv^k \rho$ and $\vv^k = \Rv \Ev^k \rho$,
\begin{align*}
\Phi_n(\uv^0, \vv^0, &\dots, \uv^{t-1}, \vv^{t-1}, \uv^t) \\
&= \Psi_n(\Hv^0, \Ev^0, \dots, \Hv^{t-1}, \Ev^{t-1}, \Hv^t) \rho \\
&\stackrel{\mathrm{P}}{\simeq} \Psi_n(\Hv^0, \zetav_\sigma^0, \dots, \zetav_\tau^{t-1}, \zetav_\sigma^{t-1}, \zetav_\tau^t) \rho \\
&= \Phi_n(\Hv^0 \rho, \Rv \zetav_\sigma^0 \rho, \dots, \zetav_\tau^{t-1} \rho, \Rv \zetav_\sigma^{t-1} \rho, \zetav_\tau^t \rho).
\end{align*}
By Lemma~\ref{lem:CollapsedState}, $\Rv \zetav_\sigma^s \rho$ and $\zetav_\tau^s \rho$ are proper complex Gaussian random variables with covariances $(\sigma^{s,r})$ and $(\tau^{s,r})$, i.e., they are distributed as $\Zv_\sigma^s$ and $\Zv_\tau^s$ respectively.
\end{IEEEproof}

\subsection{Proof of Corollary~\ref{cor:complex-amp-vec-est}}
\label{sec:amp-algo-corollary}
\begin{IEEEproof}
First, recall the AMP algorithm
\begin{align*}
    \zv^t &= \yv-\Av \xv^t + \frac{1}{\delta} \zv^{t-1} \langle (\eta^{t-1})'(\rv^{t-1}) \rangle, \\
    \rv^t &= \xv^t + \Av^{\hermitian} \zv^t, \quad
    \xv^{t+1} = \eta^t(\rv^t).
\end{align*}
    Consider the following with initialization $\uv^0 = - \xv_0$:
    \begin{align*}
        \uv^{t+1} &= \xv_0 - \rv^t, \quad
        \vv^t = \wv - \zv^t, \\
        \pzf^t(\uv) &= \eta^{t-1}(\xv_0 - \uv) - \xv_0, \quad
        \pzg^t(\vv) = \vv - \wv.
    \end{align*}
    Since $\eta^t$ is pseudo-Lipschitz by assumption, $\pzf^t(\uv) = \eta^{t-1}(\xv_0 - \uv) - \xv_0$ is pseudo-Lipschitz (composition with a translation), and $\pzg^t(\vv) = \vv - \wv$ is affine, hence Lipschitz.
    These functions satisfy Assumption~\ref{assumption:reg-complex}, so Theorem~\ref{thm:complex-amp-main} applies.
    Clearly $\pzf^t(\uv^{t}) = \eta^{t-1}(\rv^{t-1})-\xv_0 = \xv^{t}-\xv_0$, and so
    \begin{align*}
        &\lim_{N\to \infty} \frac{1}{N} \norm{\xv^t-\xv_0}_2^2 = \lim_{N\to \infty} \frac{1}{N} \langle \pzf^t(\uv^t), \pzf^t(\uv^t) \rangle \\
        & \hspace{10pt} \stackrel{{\rm Thm.~\ref{thm:complex-amp-main}}}{=} \lim_{N\to \infty} \frac{1}{N} \Eb[\langle \pzf^t(\Zv_\tau^t), \pzf^t(\Zv_\tau^t) \rangle] \\
        & \hspace{20pt} = \lim_{N\to \infty} \frac{1}{N} \Eb[\norm{\eta^{t-1}(\xv_0 - \Zv_\tau^t)-\xv_0}_2^2] \\
        & \hspace{20pt} = \lim_{N\to \infty} \frac{1}{N} \Eb[\norm{\eta^{t-1}(\xv_0 + \sqrt{\tau^{t,t} }\Zv)-\xv_0}_2^2] = \delta \sigma^{t,t}.
    \end{align*}
    But, on the other hand, we have that
    \begin{align*}
        \tau^{t+1,t+1} &= \lim_{n\to \infty} \frac{1}{n} \Eb[\langle \pzg^t(\Zv_\sigma^t),  \pzg^t(\Zv_\sigma^t)\rangle] \\
        &= \lim_{n\to \infty} \frac{1}{n} \Eb[ \norm{\pzg^t(\Zv_\sigma^t)}_2^2] = \lim_{n\to \infty} \frac{1}{n} \Eb[ \norm{\Zv_\sigma^t - \wv}_2^2] \\
        &= \lim_{n\to \infty} \frac{1}{n} \left(\Eb[ \norm{\Zv_\sigma^t}_2^2] - 2\Re\,\Eb[\langle \Zv_\sigma^t, \wv \rangle] + \Eb[\norm{\wv}_2^2]\right) \\
        &= \sigma^{t,t} + \sigma_w^2,
    \end{align*}
where the cross-term vanishes because $\Zv_\sigma^t$ is zero-mean and independent of $\wv$.
Thus, $\lim_{N\to \infty} \frac{1}{N} \norm{\xv^t-\xv_0}_2^2 = \delta(\tau^{t+1, t+1} - \sigma_w^2)$.
\end{IEEEproof}

\section{Matrix Extension: Derivations}
\label{sec:complex_matrix_amp}

This appendix provides the detailed derivations supporting the matrix extension of Section~\ref{sec:matrix-extension}.
We first describe the real-valued expansion of the observation model, then construct the real matrix-valued denoisers, and finally prove Proposition~\ref{prop:mat-red-dynamics}.

\subsection{Problem Setup and Notation}
\label{sec:mat-amp-setup}

Consider the observation model $\Yv = \Av \Xv + \Nv$, where $\Xv \in \Cb^{N\times p}$ is a complex matrix of $p$ features, each of dimension $N$, $\Av \in \Cb^{n\times N}$ is a complex Gaussian random matrix with i.i.d.\ entries $\Av_{\iota,j} \sim \cC\cN(0, 1/n)$, and $\Nv \in \Cb^{n\times p}$ has entries $\Nv_{\iota,j} \sim \cC\cN(0, \sigma_w^2)$.
We denote $\Av_{\iota,:} \in \Cb^{1\times N}$ the $\iota$th row and $\Av_{:,j} \in \Cb^{n\times 1}$ the $j$th column of $\Av$.

Define the real-valued expansions of $\Wv$ as in \eqref{eq:CanonicalMap1} and $\Sv$ as:
\begin{align*}
    \Sv &=
    \begin{bmatrix}
        \Re(\Xv_{:,1}), \Im(\Xv_{:,1}), \cdots, \Re(\Xv_{:,p}), \Im(\Xv_{:,p})
    \end{bmatrix} \in \Rb^{N\times 2p}.
\end{align*}
The product $\Pv = \Wv \Sv \in \Rb^{2n\times 2p}$ admits the block decomposition
\begin{align*}
    \Pv &= \begin{bmatrix}
        \Pv^{1,1} & \cdots & \Pv^{1,p} \\
        \vdots & \ddots & \vdots \\
        \Pv^{n,1} & \cdots & \Pv^{n,p}
    \end{bmatrix},
\end{align*}
where each block is given by
\begin{align*}
    \Pv^{\iota,j} &= \begin{bmatrix}
        \Re(\Av_{\iota,:})\Re(\Xv_{:,j}) &  \Re(\Av_{\iota,:})\Im(\Xv_{:,j}) \\
        \Im(\Av_{\iota,:})\Re(\Xv_{:,j}) & \Im(\Av_{\iota,:})\Im(\Xv_{:,j})
    \end{bmatrix} \in \Rb^{2\times 2}.
\end{align*}
Recalling $\rho = [1,~\pzi]^\transpose$ from Section~\ref{sec:canonical-transform}, we verify that
\begin{align*}
    &\rho^\transpose \Pv^{\iota,j} \rho = \begin{bmatrix}
        1 & \pzi
    \end{bmatrix}
    \begin{bmatrix}
        \Re(\Av_{\iota,:}) (\Re(\Xv_{:,j}) + \pzi \Im(\Xv_{:,j})) \\
        \Im(\Av_{\iota,:})(\Re(\Xv_{:,j}) + \pzi \Im(\Xv_{:,j}))
    \end{bmatrix} \nonumber \\
    &= (\Re(\Av_{\iota,:}) + \pzi \Im(\Av_{\iota,:}))(\Re(\Xv_{:,j}) + \pzi \Im(\Xv_{:,j})) = \Av_{\iota,:} \Xv_{:,j}.
\end{align*}
This confirms that the real expansion preserves the algebraic structure of the complex product under collapse.
One readily verifies that $\Rv_R \Wv = \Av$, $\Sv \Rv_L = \Xv$, and therefore $\Rv_R \Pv \Rv_L = \Av \Xv$.

\subsection{Real-Valued Denoiser Construction}
\label{sec:mat-amp-denoisers}

For brevity, we suppress the iteration superscript~$k$ in the remainder of this appendix.
Let $\Ev \in \Rb^{2n \times 2p}$ and $\Hv \in \Rb^{N \times 2p}$ be real matrices satisfying $\Rv_R \Ev \Rv_L = \Vv$ and $\Hv \Rv_L = \Uv$, respectively.
Define the real matrix-valued function $f$ from $\pzf$:
\begin{align*}
    \Mv {=} f(\Hv) {=}
    \begin{bmatrix}
        \Re(\pzf_{:,1}), \Im(\pzf_{:,1}), \dots, \Re(\pzf_{:,p}),\Im(\pzf_{:,p})
    \end{bmatrix} \in \Rb^{N\times 2p},
\end{align*}
so that $\Fv = \pzf(\Uv) = \Mv \Rv_L$.
Similarly, for $1\le \iota \le n$ and $1\le j\le p$, define the real matrix-valued function $g$ from $\pzg$ by
\begin{align*}
    \Qv &= g(\Ev) =
    \begin{bmatrix}
        g^{\iota,j}(\Ev)
    \end{bmatrix} \in \Rb^{2n\times 2p}, \\
    \Qv^{\iota,j} &= g^{\iota,j}(\Ev) = \begin{bmatrix}
        \Re\pzg_{\iota,j}(\Rv_R \Ev \Rv_L) & \Im \pzg_{\iota,j}(\Rv_R \Ev \Rv_L) \\
        \Im \pzg_{\iota,j}(\Rv_R \Ev \Rv_L) & -\Re \pzg_{\iota,j}(\Rv_R \Ev \Rv_L)
    \end{bmatrix}.
\end{align*}
The $2\times 2$ block structure of $\Qv^{\iota,j}$ encodes complex multiplication by $\overline{\rho}$, paralleling the vector case of Definition~\ref{def:canonical-transformation}.

\subsection{Proof of Proposition~\ref{prop:mat-red-dynamics}}
\label{sec:mat-amp-reduction}

We show that the real matrix AMP equations~\eqref{eq:real-mat-1}--\eqref{eq:real-mat-2}, when collapsed via $\Rv_R$ and $\Rv_L$, recover the complex matrix AMP equations~\eqref{eq:complex-mat-amp-U}--\eqref{eq:complex-mat-amp-V}.
The argument proceeds in three steps, verifying in turn the forward-channel terms and both Onsager correction terms.

\emph{Step~1: We show that $\Wv^\transpose \Qv \Rv_L = \Av^\hermitian \Gv$.}
From the definitions, $\Rv_R \Wv \Mv \Rv_L = \Av \Fv$.
For the adjoint term, we compute
\begin{equation} \label{eq:mat-QRL}
    \begin{split}
    (\Qv \Rv_L)^{\iota,j} &= (g(\Ev) \Rv_L)^{\iota,j} =
    \begin{bmatrix}
        \Re(\pzg_{\iota,j}(\Vv)) + \pzi \Im(\pzg_{\iota,j}(\Vv)) \\
        \Im(\pzg_{\iota,j}(\Vv)) - \pzi \Re(\pzg_{\iota,j}(\Vv))
    \end{bmatrix} \\
    &= \begin{bmatrix}
       \pzg_{\iota,j}(\Vv) \\
        -\pzi\pzg_{\iota,j}(\Vv)
    \end{bmatrix} = \pzg_{\iota,j}(\Vv) \overline{\rho}.
    \end{split}
\end{equation}
Stacking over all $\iota \in [n]$ and $j \in [p]$, we obtain $\Qv \Rv_L = \pzg(\Vv) \otimes \overline{\rho} = (\Iv_n \otimes \overline{\rho}) \pzg(\Vv) = \Rv_R^\hermitian \Gv$.
Since $\Wv$ is real, $\Wv^\transpose \Rv_R^\hermitian = \Wv^\hermitian \Rv_R^\hermitian = (\Rv_R \Wv)^\hermitian = \Av^\hermitian$, and therefore
\begin{align*}
    \Wv^\transpose \Qv \Rv_L = \Wv^\transpose \Rv_R^\hermitian \Gv = \Av^\hermitian \Gv.
\end{align*}

\emph{Step~2: We show that $\Rv_R \Qv \Bv^\transpose \Rv_L = \Gv \Lv^\transpose$, establishing $\Vv = \Av \Fv - \Gv \Lv^\transpose$.}
We first compute
\begin{equation} \label{eq:mat-RQBL}
    \begin{split}
    &(\Rv_R\Qv)^{\iota,j} = \rho^\transpose \Qv^{\iota,j} \\
    &= \begin{bmatrix}
        \Re(\pzg_{\iota,j}(\Vv))+\pzi \Im(\pzg_{\iota,j}(\Vv)), \Im(\pzg_{\iota,j}(\Vv))-\pzi \Re(\pzg_{\iota,j}(\Vv))
    \end{bmatrix} \\
    &=\pzg_{\iota,j}(\Vv) \rho^\hermitian,
    \end{split}
\end{equation}
which gives $\Rv_R\Qv = \pzg(\Vv) \otimes \rho^\hermitian = \Gv \Rv_L^\hermitian$.
It remains to show that $\Rv_L^\hermitian \Bv^\transpose \Rv_L = \Lv^\transpose$, i.e., that $\rho^\hermitian (\Bv^\transpose)^{\iota,j} \rho = \Lv^\transpose_{\iota,j}$ for each $\iota, j \in [p]$.

Recall that in the real system, $\Bv^\transpose_{\iu,\ju} = \frac{1}{m} \sum_{d=1}^{M} \frac{\partial f_{d,\iu} }{\partial \Hv_{d,\ju}}$ with $M = N$ and $m = 2n$.
Each pair $\iu, \ju$ belongs to the first or second position within the $(\iota,j)$th sub-block, where $\iota,j \in [p]$.
Expanding the sub-block, we obtain
\begin{align*}
    (\Bv^\transpose)^{\iota,j} &= \frac{1}{2n} \sum_{d=1}^{N}
    \begin{bmatrix}
        \frac{\partial \Re(\pzf_{d,\iota})}{\partial \Re(\Uv_{d,j})} & \frac{\partial \Im(\pzf_{d,\iota})}{\partial \Re(\Uv_{d,j})} \\
        \frac{\partial \Re(\pzf_{d,\iota})}{\partial \Im(\Uv_{d,j})} & \frac{\partial \Im(\pzf_{d,\iota})}{\partial \Im(\Uv_{d,j})}
    \end{bmatrix} \\
    &= \frac{1}{2n} \sum_{d=1}^{N}
    \begin{bmatrix}
        \frac{\partial }{\partial \Re(\Uv_{d,j})} \\
         \frac{\partial}{\partial \Im(\Uv_{d,j})}
    \end{bmatrix}
    \begin{bmatrix}
        \Re(\pzf_{d,\iota}) & \Im(\pzf_{d,\iota})
    \end{bmatrix}.
\end{align*}
By Lemma~\ref{lem:prop}\ref{lem:hu}, $\frac{1}{2}\rho^\hermitian \begin{bsmallmatrix} \partial / \partial \Re(\Uv_{d,j}) \\ \partial / \partial \Im(\Uv_{d,j}) \end{bsmallmatrix} = \frac{\partial}{\partial \Uv_{d,j}}$, so
\begin{align*}
    \rho^\hermitian (\Bv^\transpose)^{\iota,j} \rho = \frac{1}{n} \sum_{d=1}^N \frac{\partial \pzf_{d,\iota}}{\partial \Uv_{d,j}} = \Lv^\transpose_{\iota,j}.
\end{align*}
Therefore, $\Rv_R \Qv \Bv^\transpose \Rv_L = \Gv \Rv_L^\hermitian \Bv^\transpose \Rv_L = \Gv \Lv^\transpose$, which establishes
\begin{align*}
    \Vv = \Rv_R \Ev \Rv_L = \Rv_R (\Wv \Mv - \Qv \Bv^\transpose) \Rv_L = \Av \Fv - \Gv \Lv^\transpose.
\end{align*}

\emph{Step~3: We show that $\Cv^\transpose \Rv_L = \Rv_L \Kv^\transpose$, establishing $\Uv = \Av^\hermitian \Gv - \Fv \Kv^\transpose$.}
In the real system, $\Cv^\transpose_{\iu, \ju} = \frac{1}{m} \sum_{\du=1}^{m} \frac{\partial g_{\du,\iu}}{ \partial \Ev_{\du, \ju}}$.
Grouping into sub-blocks indexed by $\iota,j \in [p]$,
\begin{align*}
    \Cv^\transpose_{\iu, \ju} = \frac{1}{2n} \sum_{d=1}^n \sum_{u=1,2} \frac{\partial }{\partial (\Ev^{d,j}_{u,:})^\transpose} ( g_{u,:}^{d,\iota} ),
\end{align*}
where $\iu, \ju$ belong to the first or second position within the $(\iota,j)$th sub-block.
Expanding the full sub-block and substituting the block structure of $g$,
\begin{align}
     (\Cv^\transpose)^{\iota, j} &= \frac{1}{2n}
     \sum_{d=1}^n
     \begin{bmatrix}
         \frac{\partial}{\partial \Ev^{d,j}_{1,1}} \\
         \frac{\partial}{\partial \Ev^{d,j}_{1,2}}
     \end{bmatrix}
     \begin{bmatrix}
         g^{d,\iota}_{1,1} &
         g^{d,\iota}_{1,2}
     \end{bmatrix} \nonumber \\
     & \hspace{60pt} + \begin{bmatrix}
         \frac{\partial}{\partial \Ev^{d,j}_{2,1}} \\
         \frac{\partial}{\partial \Ev^{d,j}_{2,2}}
     \end{bmatrix}
     \begin{bmatrix}
         g^{d,\iota}_{2,1} &
         g^{d,\iota}_{2,2}
     \end{bmatrix} \nonumber \\
     &= \frac{1}{2n}
     \sum_{d=1}^n
     \begin{bmatrix}
         \frac{\partial}{\partial \Ev^{d,j}_{1,1}} & \frac{\partial}{\partial \Ev^{d,j}_{2,1}} \\
         \frac{\partial}{\partial \Ev^{d,j}_{1,2}} & \frac{\partial}{\partial \Ev^{d,j}_{2,2}}
     \end{bmatrix}
     \begin{bmatrix}
          \Re \pzg_{d,\iota} & \Im \pzg_{d,\iota} \\
         \Im \pzg_{d,\iota} & -\Re \pzg_{d,\iota}
     \end{bmatrix}. \label{eq:mat-CRL}
\end{align}
Multiplying on the right by $\Rv_L$ and noting that $\begin{bsmallmatrix} \Re \pzg_{d,\iota} & \Im \pzg_{d,\iota} \\ \Im \pzg_{d,\iota} & -\Re \pzg_{d,\iota} \end{bsmallmatrix} \rho = \overline{\rho} \, \pzg_{d,\iota}(\Vv)$, we obtain
\begin{equation*}
    (\Cv^\transpose \Rv_L)^{\iota,j} = \frac{1}{2n}
     \sum_{d=1}^n
     \begin{bmatrix}
         \frac{\partial}{\partial \Ev^{d,j}_{1,1}} & \frac{\partial}{\partial \Ev^{d,j}_{2,1}} \\
         \frac{\partial}{\partial \Ev^{d,j}_{1,2}} & \frac{\partial}{\partial \Ev^{d,j}_{2,2}}
     \end{bmatrix} \overline{\rho} \, \pzg_{d,\iota}(\Vv).
\end{equation*}
Since $\Re(\Vv_{d,j}) = \Ev^{d,j}_{1,1}-\Ev^{d,j}_{2,2}$ and $\Im(\Vv_{d,j}) = \Ev^{d,j}_{1,2} + \Ev^{d,j}_{2,1}$, by the chain rule the derivative matrix transforms as
\begin{align*}
     \begin{bmatrix}
         \frac{\partial}{\partial \Ev^{d,j}_{1,1}} & \frac{\partial}{\partial \Ev^{d,j}_{2,1}} \\
         \frac{\partial}{\partial \Ev^{d,j}_{1,2}} & \frac{\partial}{\partial \Ev^{d,j}_{2,2}}
     \end{bmatrix}
     \overline{\rho}
     &=
     \begin{bmatrix}
         \frac{\partial}{\partial \Re(\Vv_{d,j})} & \frac{\partial}{\partial \Im(\Vv_{d,j})} \\
         \frac{\partial}{\partial \Im(\Vv_{d,j})} &
         -\frac{\partial}{\partial \Re(\Vv_{d,j})}
     \end{bmatrix}
     \begin{bmatrix}
         1 \\
        -\pzi
     \end{bmatrix} \\
     &= \rho \left( \frac{\partial}{\partial \Re(\Vv_{d,j})}-\pzi\frac{\partial}{\partial \Im(\Vv_{d,j})} \right),
\end{align*}
where the last step uses the factorization into a scalar times~$\rho$.
By Lemma~\ref{lem:prop}\ref{lem:ev}, the Wirtinger derivative gives
\begin{align*}
    (\Cv^\transpose \Rv_L)^{\iota,j} &{=} \rho \cdot \frac{1}{n}
     \sum_{d=1}^n
    \frac{1}{2} \left( \frac{\partial}{\partial \Re(\Vv_{d,j})}-\pzi\frac{\partial}{\partial \Im(\Vv_{d,j})} \right) \pzg_{d,\iota}(\Vv) \\
    &= \rho \cdot \frac{1}{n} \sum_{d=1}^n \frac{\partial \pzg_{d,\iota}(\Vv)}{\partial \Vv_{d,j}} = \rho \Kv^\transpose_{\iota,j}.
\end{align*}
Therefore, $\Cv^\transpose \Rv_L = \Kv^\transpose \otimes \rho = (\Iv_p \otimes \rho) \Kv^\transpose = \Rv_L \Kv^\transpose$, and we conclude
\begin{align*}
    \Uv &= \Hv \Rv_L = (\Wv^\transpose \Qv - \Mv \Cv^\transpose) \Rv_L \\
    &= \Av^\hermitian \Gv - \Mv \Rv_L \Kv^\transpose = \Av^\hermitian \Gv - \Fv \Kv^\transpose. \qedhere
\end{align*}

\bibliographystyle{IEEEbib}
\bibliography{refs}

\end{document}